\documentclass{article}
\usepackage{fullpage}
\usepackage[dvipsnames]{xcolor}
\usepackage{hyperref}
\hypersetup{
    colorlinks=true,
    linkcolor=blue,
    citecolor=violet,
    filecolor=magenta,
    urlcolor=cyan,
}

\usepackage{authblk}
\usepackage{graphicx}
\graphicspath{{./figures/}}
\usepackage{wrapfig}

\usepackage{tikz-cd}

\usepackage{amssymb,amsmath,amsthm,bbm}
\usepackage{mathtools}%
\usepackage[capitalize]{cleveref}

\usepackage{amsthm,amssymb,amsmath}
\usepackage{cutwin}
\usepackage{enumitem}
\usepackage{microtype}

\numberwithin{equation}{section}

\newtheorem{theorem}[equation]{Theorem}
\newtheorem{corollary}[equation]{Corollary}
\newtheorem{lemma}[equation]{Lemma}
\newtheorem{remark}[equation]{Remark}

\newtheorem{proposition}[equation]{Proposition}

\newtheorem{definition}[equation]{Definition}

\Crefname{defn}{Defn.}{Defns.}
\Crefname{lemma}{Lem.}{Lems.}
\Crefname{alg}{Alg.}{Algs.}

\usepackage{thmtools}
\usepackage{thm-restate}

\usepackage[backend = biber, url=false,sorting=nyt,style=numeric-comp,maxbibnames = 99]{biblatex}
\usepackage[colorinlistoftodos,prependcaption,textsize=small]{todonotes}

\definecolor{ltgray}{gray}{0.9}
\definecolor{gray}{rgb}{0.95,0.95,0.96}
\definecolor{dkgray}{rgb}{0.7,0.7, 0.735}
\definecolor{ltblue}{rgb}{0.55,0.55, 0.95}
\definecolor{ltGreen}{rgb}{0.25,0.65, 0.25}
\definecolor{dkgreen}{RGB}{0, 100, 0}
\definecolor{dkred}{rgb}{0.75,0.0, 0.0}
\definecolor{dkorange}{rgb}{.82,.45, 0.0}
\definecolor{ltred}{rgb}{0.95,0.95, 0.85}
\definecolor{utahRed}{rgb}{.8, 0, 0}
\definecolor{oregonGreen}{rgb}{0, .41, .163}
\definecolor{albanyPurple}{rgb}{0.4, 0, .55}

 \newcommand{\fullVerRef}[1]{\cref{#1}}
\newcommand{\E}{\mathbb{E}}
\newcommand{\N}{\mathbb{N}}
\newcommand{\R}{\mathbb{R}}
\newcommand{\W}{\mathbb{W}}
\newcommand{\V}{\mathbb{V}}
\newcommand{\X}{{\mathbb X}}
\newcommand{\Y}{{\mathbb Y}}
\newcommand{\Z}{{\mathbb Z}}
\newcommand{\CC}{\mathcal{C}}
\newcommand{\DD}{\mathcal{D}}

\newcommand{\End}{\mathbf{End}}

\newcommand{\e}{\varepsilon}
\renewcommand{\phi}{\varphi}
\newcommand{\inv}{^{-1}}

\newcommand{\1}{\mathbbm{1}}
\renewcommand{\Im}{\mathrm{Im}}
\newcommand{\id}{\mathrm{Id}}
\newcommand{\Id}{\id}

\newcommand{\from}{\colon}
\newcommand{\after}{\circ}

\newcommand{\Reeb}{\ensuremath{\mathbf{Reeb}}}

\begin{document}
\title{A family of metrics from the truncated smoothing of Reeb graphs
 \thanks{EC was supported in part by NSF grants CCF-1614562, CCF-1907612, and DBI-1759807.
 EM and TO were supported in part by NSF grant CCF-1907591.
 EM was additionally supported in part by DEB-1904267.
 }
}

\author[1]{Erin Wolf Chambers}
\author[2,3]{Elizabeth Munch}
\author[2]{Tim Ophelders}
\affil[1]{Dept.~of Computer Science, St.~Louis University}
\affil[2]{Dept.~of Computational Mathematics, Science and Engineering, Michigan State University}
\affil[3]{Dept.~of Mathematics, Michigan State University}

\date{}
\maketitle              %

\begin{abstract}
In this paper, we introduce an extension of smoothing on Reeb graphs,
which we call truncated smoothing; this in turn allows us to define a new family of metrics which generalize the interleaving distance for Reeb graphs.
Intuitively, we ``chop off'' parts near local minima and maxima during the course of smoothing, where the amount cut is controlled by a parameter $\tau$.
After formalizing truncation as a functor, we show that when applied after the smoothing functor, this prevents extensive expansion of the range of the function, and yields particularly nice properties (such as maintaining connectivity) when combined with smoothing for $0 \leq \tau \leq 2\e$, where $\e$ is the smoothing parameter.
Then, for the restriction of $\tau \in [0,\e]$, we have additional structure which we can take advantage of to construct a categorical flow for any choice of slope $m \in [0,1]$.
Using the infrastructure built for a category with a flow, this then gives an interleaving distance for every $m \in [0,1]$, which is a generalization of the original interleaving distance, which is the case $m=0$.
While the resulting metrics are not stable, we show that any pair of these for $m,m' \in [0,1)$ are strongly equivalent metrics, which in turn gives stability of each metric up to a multiplicative constant.
We conclude by discussing implications of this metric within the broader family of metrics for Reeb graphs.

\end{abstract}

\section{Introduction}

The Reeb graph, originally defined in the context of Morse theory \cite{Reeb1946}, represents a portion of the underlying structure of a topological space $\X$ through the lens of a real valued function $h: \X \to \R$; the pair of data $(\X,h)$ is known as an $\R$-space.
Specifically, points in the Reeb graph correspond to connected components in the levelsets of the function; as such, the Reeb graph inherits a real valued function from the original input data.
For nice enough inputs, the resulting object is a finite graph.
So, at its core, we focus our  study on objects of the form $(G,f)$ where $G$ is a graph and $f: G \to \R$ is a  function given on vertices and interpolated linearly on the edges.
See \cref{fig:smoothingOperation} for an example.

Reeb graphs have become increasingly useful in a wide range of applications, including settings such as
shape comparison \cite{Hilaga2001, Escolano2013},
denoising \cite{Wood2004},
shape understanding \cite{Dey2013,Hetroy2003},
reconstructing non-linear 1-dimensional structure in data \cite{Natali2011,Ge2011,Chazal2014,Szymczak2011},
summarizing collections of trajectory data \cite{Buchin2013},
and allowing for informed exploration of otherwise hard-to-visualize high-dimensional data~\cite{Weber2007, Harvey2010a}; see~\cite{Biasotti2008} for a survey of these and more topics.  
As a result, there is interest in defining metrics on these objects, to evaluate their quality in the face of noisy input data as well as to allow for more accurate shape comparison and analysis.
In this setting, we are focused on metrics that incorporate both the graph and function information: so $d((G,f_1),(G,f_2))$ should be non-zero if $f_1 \neq f_2$ even though they are defined on the same underlying graphs.

\begin{figure}
    \centering
    \includegraphics[page=1]{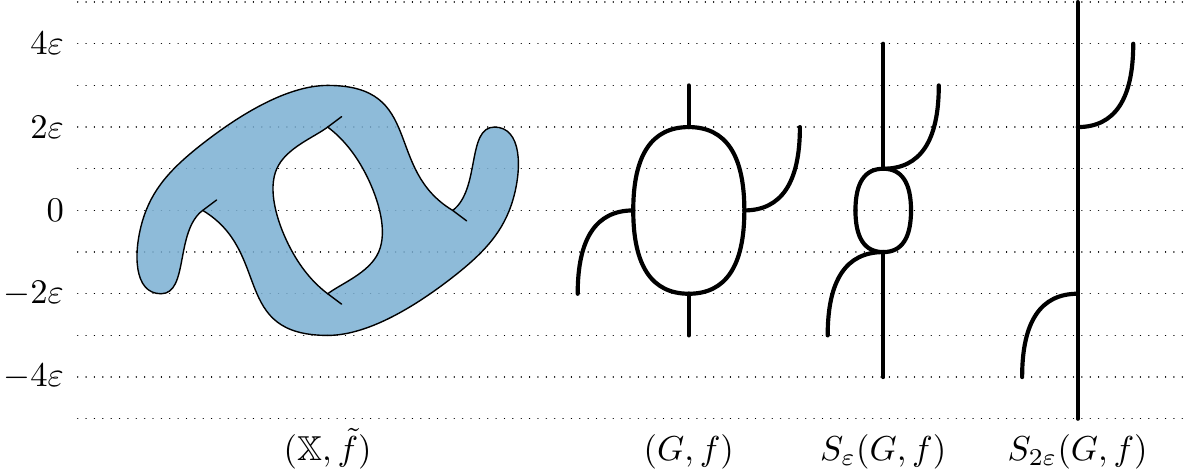}
    \caption{From left to right: an $\R$-space $(\X,\tilde f)$, its Reeb graph $(G,f)$, smoothings are shown for two parameters, $\e$ and $2\e$. Function values are shown by height.}
    \label{fig:smoothingOperation}
\end{figure}

Several metrics have arisen recently to do this, taking inspiration from different mathematical backgrounds \cite{deSilva2016,Bauer2015b,Bauer2020a,Carriere2017,DiFabio2012,DiFabio2016,Bauer2016c,Sridharamurthy2018,Bauer2014}.
In this paper, we focus on the Reeb graph interleaving distance \cite{deSilva2016}.
The basic idea is to work with a notion of smoothing, which returns a parameterized family of Reeb graphs, $S_\e(G,f)$ for every $\e \geq 0$, starting with $\e=0$ which leaves the input unchanged.
This procedure simplifies the loop structures and stretches tails \cite{Yuda2019}; see \cref{fig:smoothingOperation} for an example.
Then the goal is to find an $\e$-interleaving, which is a pair of families of maps making a particular diagram commute.
If $\e=0$, this diagram simplifies down to finding an isomorphism between the two Reeb graphs; increasing $\e$ provides more flexibility to find such pairs of maps.
Then we have a metric by defining $d_I((G,f),(H_h))$ to be the infimum over the set of $\e$ for which such a diagram exists.

This metric takes root in the interleaving distance defined for persistence modules \cite{Chazal2009b}, and is largely inspired by the subsequent category theoretic treatment \cite{Bubenik2014,Bubenik2014a}.
This viewpoint comes from encoding the data of a Reeb graph in a constructible set-valued cosheaf \cite{Curry2014,Curry2016b}.
It was later shown that these metrics are special cases of a more general theory of interleaving distances given on a \textit{category with a flow} \cite{deSilva2018,Stefanou2018, Cruz2019}.
This framework encompases common metrics including $\ell_\infty$ distance on points or functions, regular Hausdorff distance, and the Gromov-Hausdorff distance \cite{Stefanou2018,Bubenik2017a}.
Using this framework, interleaving metrics have been studied in the context of
$\R$-spaces \cite{Blumberg2017},
multiparameter persistence modules \cite{Lesnick2015},
merge trees \cite{Morozov2013}, and
formigrams \cite{Kim2017,Kim2019a},
and on more general category theoretic constructions \cite{Botnan2020,Scoccola2020a}.
There are also interesting restrictions to labeled merge trees, where one can pass to a matrix representation and show that the interleaving distance is equivalent to the point-wise $\ell_\infty$ distance
\cite{Munch2019,Gasparovic2019,Yan2019a,Stefanou2020}.

On the negative side, it has been shown that Reeb graph interleaving is graph isomorphism complete \cite{deSilva2016,Bjerkevik2017}, and that many other variants are also NP-hard \cite{Bjerkevik2017,Bjerkevik2019}.  All of this means that these metrics, while mathematically interesting, may not lead to feasible algorithms for comparison and analysis.
However, a glimmer of hope arises with work investigating fixed parameter tractable algorithms \cite{Touli2018,Stefanou2020}.
Despite the issues of computational complexity, notions of similarity for graphs in general, and Reeb graphs in particular, are of pressing interest due to their extensive use in data analysis; in many such settings, we are concerned with questions of quality in the face of noise, and understanding convergence of approximations to a true underlying structure.
For example, the interleaving distance has been used in evaluating the quality of the mapper graph \cite{Singh2007}, which can be proven to be a approximation of the Reeb graph using this metric \cite{Munch2016,Brown2019}.
Furthermore, there is considerable interest in unifying the interleaving distance with the emerging collection of other Reeb graph metrics.

In this paper, we introduce a truncation operation, which intuitively cuts off portions of the Reeb graph near local extrema with respect to $f$; this operation is easy to compute for any Reeb graph and tends to result in a simplified Reeb graph.
We show that truncation is a functor, and when combined with the smoothing functor, defines a flow on the the category of Reeb graphs.
We investigate and prove particularly desirable geometric and topological properties of \emph{truncated smoothing} for certain ranges of the two parameters controlling the functors.
We then introduce a new family of metrics for Reeb graphs, called truncated interleaving distances.
They are parameterized by $m \in [0,1]$, and generalize the interleaving distance, with the setting $m=0$ being the original interleaving distance.
We show that the metrics arising from $m \in [0,1)$ are strongly equivalent.
Although the metrics are not stable in the sense of \cite{Bauer2020a}, strong equivalence implies that they are at least stable up to a constant.

When combined with preliminary work on geometric implications of smoothing~\cite{Yuda2019}, truncated smoothing  is interesting in its own right, as it provides a collection of paths for Reeb graph space to be studied in terms of the resulting persistence diagrams. It also is useful when considering algorithms to test planarity for Reeb graphs, or find planar representations of them.
The new family of metrics also provide the possibility for new approaches for approximation algorithms for the interleaving distance, as well as
new avenues for further unification of the broader family of Reeb graph metrics.

\subparagraph{Outline}
We give the basic background on Reeb graphs, smoothing, and the Reeb graph interleaving distance in \cref{ssec:Background}, with more complete details on categories and interleavings given in \fullVerRef{sec:CategoriesBackground}. 
Next, we introduce our definition of truncated smoothing in \cref{sec:TruncatedSmoothing}; again,  some alternative formulations and full justification are in \fullVerRef{sec:EquivalentDefinitions}. 
In \cref{ssec:PropertiesOfTruncationSurvey} we check properties of the truncated smoothing operation, with complete proofs given in \cref{sec:properties_of_truncation}. 
Then we take a categorical view of truncated smoothing to develop a family of metrics and investigate their properties in \cref{ssec:truncInterleavingSurvey,ss:propertiesOfMetric}, including the relevant technicalities in \cref{sec:mapsAndProperties,sec:IsomorphismIssues,sec:categorical_flow,sec:appendix_strongequivmetrics}. 
Finally, some implications of our work as well as possible future directions are discussed in  \cref{sec:Conclusions}.

\section{Overview of results}

In this section, we give an  overview of the main definitions (both new and old), as well as stating the properties and implications of these results.
We reserve several alternate constructions, more technical background material, and many of the proofs for later sections.

\subsection{Background: Reeb graphs, smoothing, and  interleaving}
\label{ssec:Background}
Given a topological space $\X$ along with a continuous $\R$-valued function $f\from\X\to\R$, we call the pair $(\X,f)$ an \emph{$\R$-space}.
For two $\R$-spaces $(\X,f)$ and $(\X',g)$, we call a continuous map $\phi\from\X\to\X'$ \emph{function-preserving} if $f=g\circ\phi$, and write $\phi\from(\X,f)\to(\X',g)$ in that case.

For an $\R$-space $(\X,f)$, we define an equivalence relation $\sim_f$ on the points of $\X$, such that $x\sim_f x'$ if and only if $x$ and $x'$ lie in the same path-connected component of $f\inv(y)$ for some $y\in\R$.
For sufficiently nice functions\footnote{e.g.~A Morse function on a manifold, or a constructible space and function \cite{deSilva2016}, or a space with a levelset-tame function \cite{Dey2013a}.}, the quotient space $\X/\!\!\sim_f$ is a graph, called a \emph{Reeb graph}, and we denote the quotient map by $q_f\from(\X,f)\to(\X/\!\!\sim_f,g)$.
Since $f(x)=f(x')$ whenever $x\sim_f x'$, we can treat the Reeb graph as an \emph{$\R$}-space $(X/\!\!\sim_f,g)$ by defining $g(q_f(x))=f(x)$, so that $q_f$ is function-preserving.
Most but not all functions in this paper are function preserving.
\cref{fig:smoothingOperation} illustrates the construction of a Reeb graph of an $\R$-space.

For the purposes of this work, we will largely divorce the idea of the Reeb graph from the need for a starting space that was used to construct it.
Thus for our purposes, a \emph{Reeb graph} is a pair $(G,f)$ where $G=(V_G,E_G)$ is a finite multigraph and $f\from G\to\R$, referred to as the \emph{height function}, is a continuous map that is linearly interpolated along edges of $G$, and for which no two neighboring vertices have the same function value.
We write $\Im(G,f) := f(G) \subset \R$ for the image of the graph in $\R$. 
The function can equivalently be stored by defining $f\from V_G \to \R$ as a function on the vertices, and extending it to the edges implicitly.
We treat $G$ as a topological space, so that a point $x\in G$ lies either on a vertex of $G$, or interior to an edge of $G$.
For succinctness, we also write $x \in (G,f)$ to mean $x\in G$.
Since no two adjacent vertices have the same function value, a level set $f\inv(y) $ for  $y\in\R$ is a finite set of points in $G$ which could be vertices and/or points in the interior edges.

Together, the collection of Reeb graphs (treated as $\R$-spaces) with function-preserving maps as morphisms forms a category, $\Reeb$.
For the reader without a background in category theory, the basic idea is that that this collection of objects and morphisms satisfy some basic axiomatic structures that make their analysis easier to view as a collection.
It also makes available the viewpoint of \emph{functors} between categories, which are essentially structure preserving maps.
For now, we will largely hand-wave past the categorical constructions, and defer the technicalities to \cref{sec:CategoriesBackground}.

\begin{figure}
    \centering
    \includegraphics[page=2]{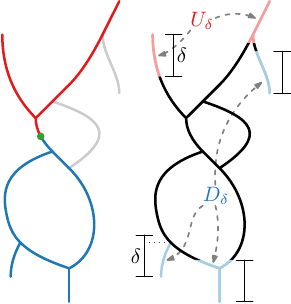}
    \caption{
    Left: the up-set (red) and down-set (blue) of a point.
    Although the up-set is a tree, it is not an up-tree as it contains down-forks of the ambient graph.
    Right: the sets $U_\delta$ and~$D_\delta$ of points with no length $\delta$ up-path or down-path, respectively.
    The leftmost component of $D_\delta$ does not contain the down-fork.
    }
    \label{fig:UpAndDownForests}
\end{figure}
Define a \emph{path} from $x$ to $x'$ in $(G,f)$ to be a continuous map $\pi\from[0,1]\to G$ such that $\pi(0)=x$ and $\pi(1)=x'$.
A path is called an \textit{up-path} if it is monotone-increasing with respect to the function, i.e.~$f(\pi(t)) \leq f(\pi(t'))$ for $t \leq t'$.
Symmetrically, a path is a \textit{down-path} if it is monotone-decreasing.
In the case of an up- or down-path $\pi$, we call $|f(\pi(0)) - f(\pi(1))|$ the \textit{height} of the path.

In a Reeb graph $(G,f)$, let the \emph{up-paths} of a point $x$ be the set of $f$-monotone paths that have $x$ as minimum.
The \emph{up-set} of a point $x$ is the set of points reachable from $x$ by an up-path, including $x$ itself.
Define an \emph{up-fork} to be a vertex $x$ whose up-set contains at least two edges adjacent to $x$.
We define \emph{down-paths}, \emph{down-sets}, and \emph{down-forks} symmetrically.
Call the up-set of a point $x$ an \emph{up-tree} if it contains no down-forks of $(G,f)$, and say that $x$ \emph{roots} an up-tree in such case.
The concept of rooting a \emph{down-tree} is defined symmetrically.
See \cref{fig:UpAndDownForests}.
    \begin{definition}\label{defn:originalSmoothing}
        Fix a  Reeb graph $(G,f)$ and $\e\geq 0$.
        Define the \emph{$\e$-thickening} of $G$ to be the space
        $G\times[-\e,\e]$
        with the product topology,
        and define
        $(f+\Id)\from G \times [-\e,\e]\to\R$ by
        $(f+\Id)(x,t)=f(x)+t$.
        We define the \emph{$\e$-smoothing} $S_\e(G,f)$ to be the Reeb graph of $(f+\Id) $, and denote the corresponding quotient map by $q \from G \times [-\e,\e] \to S_\e(G,f)$.
        The composition of $q$ with the the inclusion $G \hookrightarrow G \times [-\e,\e]; x \mapsto (x,0)$ is denoted 
        $\eta = q \circ  (\Id, 0) $.
    \end{definition}

    \begin{figure}
    \centering
        \includegraphics[page=1,width = 2.1in]{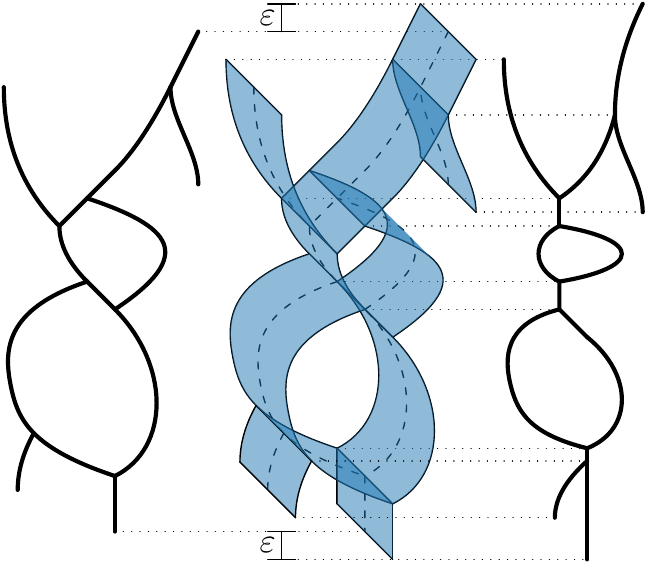}
        \caption{From left to right: a Reeb graph $(G,f)$, its $\e$-thickening $(G \times [-\e,\e],f+\id)$, and the Reeb graph $S_\e(G,f)$ of the $\e$-thickening.
        The product of an edge with an interval is drawn to reflect the function value at a given height. }%
        \label{fig:thickening}
    \end{figure}
See \cref{fig:thickening} for an example. In essence, smoothing eliminates small cycles whose height is $\le 2\e$, and shrinks all other cycles; it also moves every up-fork and local maximum up and every down-fork and local minimum down.  Under the lens of studying the topology of the graph (and in turn the original space), this serves as a functor that can be used to remove noise and simplify topology in a parameterized fashion.

The smoothing construction, $S_\e$, holds quite a bit more useful structure as not only is it a functor, it is an example of a flow \cite{deSilva2018}.
While we do not provide the full definition here, the specifics are given in \fullVerRef{ssec:categories}.
In particular, this comes from using the additional structure afforded by the function preserving map $\eta\from(G,f) \to S_\e(G,f)$.
We will reserve the full investigation of $\eta$ until \fullVerRef{ssec:Reebinterleaving}, but will use the following property of categories with a flow. 

\begin{theorem}
[{\cite[Thm.~2.7]{deSilva2018}}]
\label{thm:flowGivesMetric}
A category with a flow gives rise to an interleaving distance on the objects of the category;  specifically, this construction is an extended pseudometric.
\end{theorem}

This construction is quite useful since simply by finding some relatively easy to check structure on a category, we immediately get  a distance measure on the objects.
Depending on the category and flow, this construction encompasses many standard metrics such as the Hausdorff distance; and with a choice of other categories and flows we can construct new metrics.
We are particularly interested in the special case of the interleaving distance for Reeb graphs as studied in \cite{deSilva2016}. 

\begin{definition}
An \emph{$\e$-interleaving} with respect to $S_\e$ is a pair of maps, $\phi\from (G,f) \to S_\e(H,h)$ and $\psi\from (H,h) \to S_\e(G,f)$ such that the diagram 
\begin{equation*}
\begin{tikzcd}[column sep=5em, row sep = huge]
    (G,f) \ar[r, "\eta"] \ar[dr,OliveGreen, "\phi"', very near start]
    & S_\e(G,f) \ar[r, "{S_\e[\eta]}"] \ar[dr,OliveGreen, "{S_\e[\phi]}"',very near start, dashed]
    & S_{2\e}(G,f)
    \\
    (H,h) 
    \ar[r, "\eta"'] 
    \ar[ur,violet, "\psi", very near start]
    & S_\e(H,h) 
    \ar[r,"{S_\e[\eta]}"'] 
    \ar[ur,violet,dashed, "{S_\e[\psi]}", very near start]
    & S_{2\e}(H,h)
\end{tikzcd}
\end{equation*}
commutes.
The \emph{interleaving distance} is defined to be 
\begin{equation*}
    d_I((G,f), (H,h)) = \inf_{\e}\{\text{there exists an $\e$-interleaving of $(G,f)$ and $(H,h)$} \}.
\end{equation*}
\end{definition}
In the construction on this category, $d_I$ is an extended metric since the interleaving distance between Reeb graphs with different numbers of connected components is $\infty$ as there is no interleaving available for any $\e$ \cite{deSilva2016}.
One particularly useful property we will make use of is understanding how the image of the smoothed Reeb graph, $\Im(S_\e(G,f)) := f(G) \subseteq \R$, changes under smoothing.
Note that if $G$ is connected, $\Im(G,f)$ is connected so it is an interval.

\begin{proposition}
\label{prop:ImageOfSmoothedGraph}
For a connected Reeb graph $(G,f)$ with 
$\Im(G,f) = [a,b]$, 
\[\Im(S_\e(G,f)) = [a-\e,b+\e].\]
\end{proposition}

\begin{proof}
For any $c\in \Im(S_\e(G,f))$, we show that $c\in[a-\e,b+\e]$.
There is some $x \in S_\e(G,f)$ with $f_\e(x) = c$, where $f_\e$ is the induced function on $S_\e(G,f)$.
Then there is a $(y,t) \in G \times [-\e,\e]$ with $f(y) + t = c$.
Combining $a \leq f(y) \leq b$ and $-\e \leq t \leq \e$ gives that $a-\e \leq c \leq b+\e$.

For the other direction, let $c \in [a-\e, b+\e]$.
There exists some $d\in[a,b]$ with $c-d\in[-\e,\e]$.
Because $\Im(G,f)=[a,b]$, there exists some $x\in f\inv(d)$ and $(x,c-d)\in G\times[-\e,\e]$ quotients to some $y\in S_\e(G,f)$ with $f_\e(y)=c$, so $\Im(S_\e(G,f))=[a-\e,b+\e]$.
\end{proof}

\subsection{Truncated smoothing}
\label{sec:TruncatedSmoothing}
We can now introduce our new, modified smoothing of Reeb graphs.
Notice  from \cref{prop:ImageOfSmoothedGraph} that as the Reeb graph is smoothed, the image becomes larger.
The basic idea of truncated smoothing is to cut off some of those expanding tails in a well-defined way.

Let $U_\tau(G,f)$ be the set of points of $G$ that do not have a length $\tau$ up-path, and define $D_\tau(G,f)$ symmetrically for down-paths.
Note that for any point $x\in U_\tau(G,f)$, all up-paths from $x$ also lie in $U_\tau(G,f)$; the symmetric property is true for $D_\tau(G,f)$.
Both $U_\tau(G,f)$ and $D_\tau(G,f)$ are open subsets of $(G,f)$.
See \cref{fig:UpAndDownForests} for an example.
With this, we can define truncation as follows.

\begin{definition}
\label{defn:truncation}
    The \emph{$\tau$-truncation} of $(G,f)$,
    is the subgraph of $(G,f)$ consisting of the points that have both an up-path and a down-path of height $\tau$; specifically
    \[T^\tau(G,f):=(G,f)\setminus(U_\tau(G,f)\cup D_\tau(G,f)).\]
\end{definition}
This operation can be seen in the second and third graphs of \cref{fig:smoothingOperation_withTrunc}.
Notice that $T^0(G,f) = (G,f)$, and that for large enough $\tau$, it is entirely possible to disconnect the graph, or even to be left with an empty graph.
Utilizing the truncation operation in conjunction with the Reeb graph smoothing operation is what we call truncated smoothing.

\begin{definition}
\label{defn:TruncatedSmoothing}
Let $(G,f)$, $\e \geq 0$ and $\tau \geq 0$ be given.
Then the \emph{truncated smoothing} of $(G,f)$ is defined by
$S_\e^\tau(G,f) =  T^\tau S_\e(G,f)$.
\end{definition}

If $\tau = 0$, $S_\e^0(G,f) = T^0 (S_\e (G,f)) = S_\e(G,f)$.
So $S_\e^0$ is the same as $S_\e$,  and thus the truncated smoothing can be thought of as a generalization of the smoothing definition.

    \begin{figure}
        \centering
        \includegraphics[page=2, width = \textwidth]{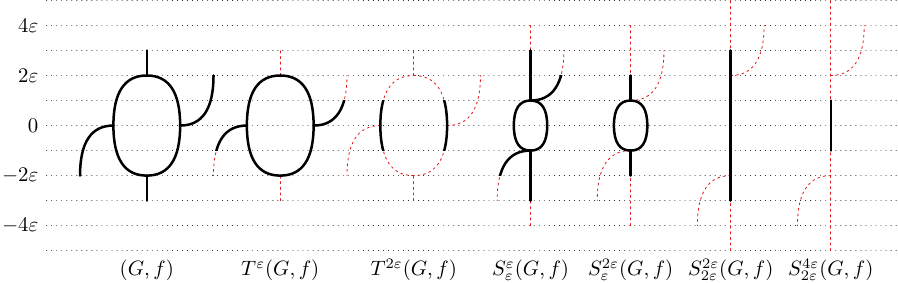}
        \caption{Example of smoothing and truncating for a range of values, on the graph from \cref{fig:smoothingOperation}.}
        \label{fig:smoothingOperation_withTrunc}
    \end{figure}

Consider~\cref{fig:smoothingOperation_withTrunc}, which shows why we smooth before truncating and more generally, why we will soon want to place restrictions on the relationship between $\tau$ and $\e$.
    Namely, for this example, we have drawn $T^\e(G,f)$ and $T^{2\e}(G,f)$.
    In the second case in particular, it is clear that truncation has massive detrimental effects on the topology as evidenced by the fact that $T^{2\e}(G,f)$ has two connected components.
    However, we can avoid these issues when we smooth first.
    In the last four examples, smoothing serves to move cycles away from the extrema, so that for a limited amount of truncation, no cycles are broken.
    We will quantify this `safe' amount of truncation in \cref{ssec:PropertiesOfTruncationSurvey}.
    So, while the smoothing parameter still gets rid of the center circle, the truncation only gets rid of expanding tails.

\subparagraph{Algorithm}
The $\tau$-truncation of a Reeb graph $(G,f)$ can be computed by first storing the length of the longest up-path and down-path of each vertex.
This can be done in linear time using a topological sort of the graph based on directing all edges upward.
We can for each local maximum store that it has a $0$-length up-path, and for the remaining vertices,  processes  in the order given by the topological sort, storing the length of their up-path based on the stored length of all previously processed neighbors.
We store the length of the longest down-path for each vertex symmetrically.
Now, we can compute for each edge how much of it remains in the truncation, and subdivide the edges if necessary.
Finally, remove all vertices and edges that do not have a sufficiently long up-path or down-path.
This procedure takes $O(n+m)$ time on a graph with $n$ vertices at $m$ edges.
The truncated smoothing can be computed by first computing the smoothing~\cite{deSilva2016} in $O(m\log(m+n))$ time, giving a total running time of $O(m\log(m+n))$.

\subsubsection{Properties of truncated smoothing}
\label{ssec:PropertiesOfTruncationSurvey}

\begin{figure}
        \centering
        \includegraphics[page=3, width = 2in]{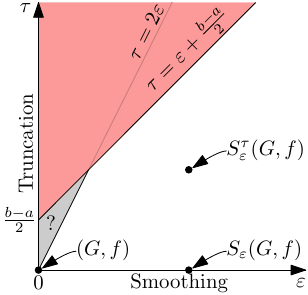}
        \caption{
        Visualization of \cref{prop:ImageOfSmoothingAndTruncating}.
        Given a connected $G$ where $\Im(G,f) = [a,b] \subset \R$,  $S_\e^\tau(G,f)$ is empty if it is in the red region and non-empty if it is in the white region. Parameters in the grey region can be either empty or not.
        }
        \label{fig:EmptyParameters}
\end{figure}
We can visualize the relationship between $\tau$ and $\e$ as drawn in \cref{fig:EmptyParameters}.
For this figure, we assume we start with a connected Reeb graph $(G,f)$ and study properties of $S_\e^\tau(G,f)$ which is represented by the point $(\e,\tau)$ in the plane.
In the remainder of this section, we state the properties of $S_\e^\tau$ in different regions of the $\e$-$\tau$-plane, culminating in the parameter space labeling of \cref{fig:epstauproperties}.
We will focus in this section on the case where $G$ is a connected graph, although some results can be modified to incorporate disconnected inputs.
These results on disconnected graphs, as well as many of the more technical proofs, are presented in  \fullVerRef{sec:properties_of_truncation}.

\subsubsection{When is \texorpdfstring{$S_\e^\tau(G,f)$}{Set(G,f)} empty? }

We first study the values of $\e$ and $\tau$ for which the truncated smoothing is empty.
For the purposes of notation, define $\Im(G,f) = f(G) \subset \R$.
Consider the following simple example:
Let $L_{[a,b]}$ be a Reeb graph consisting of a single edge with image $[a,b] \subseteq \R$, and for an interval $I\subseteq [a,b]$, let $L_I\subseteq L_{[a,b]}$ be the unique subgraph with image $I$.
Then $T^\tau(L_{[a,b]})=L_{[a+\tau,b-\tau]}$ if $2\tau\leq b-a$, and is the empty Reeb graph for $2\tau>b-a$.
On the other hand, $S_\e(L_{[a,b]})$ is isomorphic to $L_{[a-\e,b+\e]}$.

In particular, $T^\tau$ and $S_\e$ transform any monotone path with image $[a,b]$ into a monotone path with image $[a+\tau,b-\tau]$, $[a-\e,b+\e]$, respectively.
In addition, smoothing or truncating the empty Reeb graph again yields the empty Reeb graph.
We can build this intuition into the following proposition; details are in \fullVerRef{ssec:apdx:Empty}.
Note that in the case of a connected graph $G$, $\Im(G,f)$ is connected and thus is an interval.

\begin{restatable}{proposition}{ImageOfSmoothingAndTruncating}
\label{prop:ImageOfSmoothingAndTruncating}
Let $(G,f)$ be connected with $\Im(G,f) = [a,b]$. 
\begin{itemize}
    \item If $b-a < 2(\tau-\e) $, then $\Im(S_\e^\tau(G,f)) =\emptyset$.
    \item If $b-a \geq 2(\tau-\e) $ and $\tau\leq 2\e$, then 
    $\Im(S_\e^\tau(G,f)) = [a- (\e - \tau), b+ (\e-\tau)]$.
\end{itemize}
\end{restatable}

\begin{proof}
[Sketch proof.]
We first show that $b-a<2(\tau-\e)$ implies the image is empty.
We show in \cref{prop:ImageOfTruncatedGraph} that for a connected graph $(H,h)$ with image $[a',b']$ and  $b'-a' < 2\tau$, $T^\tau(H,h)$ is empty. 
    By \cref{prop:ImageOfSmoothedGraph}, $\Im(S_\e(G,f)) = [a-\e, b+\e]$.
    Then setting $S_\e(G,f) = (H,h)$, we have for $b-a < 2(\tau-\e)$, that $(b+\e) - (a-\e) \leq 2\tau$, so
    $\Im(S_\e^\tau(G,f)) = \Im(T^\tau(S_\e(G,f))) = \emptyset$.

    Now, we can assume $b-a \geq 2(\tau-\e)$.
    One direction of containment is easy since by \cref{prop:ImageOfSmoothedGraph},
    $
        \Im(S_\e^\tau(G,f)) = \Im(T^\tau(S_\e(G,f))) \subseteq [a- (\e - \tau), b+ (\e-\tau)].
        $
    Thus, it remains to show that $[a- (\e - \tau), b+ (\e-\tau)] \subseteq \Im(S_\e^\tau(G,f))$.
    The basic idea is to take two points $s,t\in S_\e(G,f)$ with $f(s)=a-\e$ and $f(t)=b+\e$, and show that they are connected by a path $\pi$ in $S_\e(G,f)$ for which the only portions that get truncated are the endpoints.
    This is simple if $\pi$ is itself a monotone path; otherwise we use the fact that $G$ has already been smoothed and that we do not truncate too much ($\tau \leq 2\e$) to show that the parts of the path which are not monotone still have long enough up- and down-paths to not be removed.
\end{proof}

This proposition gives us that $S_\e^\tau(G,f)$ is an empty graph if $(\e,\tau)$ is interior to the red region of \cref{fig:EmptyParameters}, and is empty in the white region.
We cannot expand this proposition to the grey region of \cref{fig:EmptyParameters} 
as there are examples for which $S_\e^\tau(G,f)$ can be either empty or not.  
For instance, in the example of \cref{fig:EmptyTauTruncation},  $|\Im(G,f)| \geq  2(\tau-\e)$, but each position in the graph
is either missing a long enough up- or down-path,
and hence the truncated graph is empty.
On the other hand, for the graph with a single edge $L_{[a,b]}$, any truncation $\tau < \tfrac{b-a}{2}$ is non empty. 

    \begin{figure}
        \centering
        \includegraphics{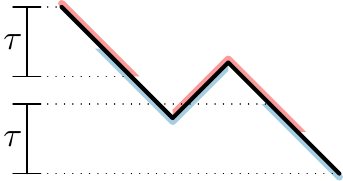}
        \caption{A Reeb graph $(G,f)$ for which $T^\tau(G,f) = S_0^\tau(G,f)$ is empty. This choice of $\tau$ is such that $\Im(G,f)$ has diameter greater than $2\tau$, thus $S_0^\tau(G,f)$ is in the grey region of \cref{fig:EmptyParameters}.} 
        \label{fig:EmptyTauTruncation}
    \end{figure}

\subsubsection{When does \texorpdfstring{$S_\e^\tau(G,f)$}{Set(G,f)} maintain connectivity?}
Our next goal is to understand when truncation preserves the connectivity of the input.
As seen in \cref{fig:smoothingOperation_withTrunc}, clearly just truncating the graph can disconnect an originally connected graph.
However, what is interesting is that smoothing first and not truncating too much relative to the smoothing will maintain the connectivity; this will be made precise in  \cref{cor:ConnectedSmoothingAndTruncating}.
For this, we introduce two properties, \emph{$t$-tailed} and \emph{$s$-safe}, and study how they are affected by smoothing and truncation.

 \begin{definition}
\label{defn:tauTailed_and_tauSafe}
    A Reeb graph is \emph{$t$-tailed} if it has a height $t$ up-path at every down-fork and a length $t$ down-path at every up-fork.
    A Reeb graph is \emph{weakly $s$-safe} if each component has a point with both an up-path and a down-path of height at least~$s$.
    A Reeb graph is \emph{$s$-safe} if it is both $s$-tailed and weakly $s$-safe.
\end{definition}

Note that every non-empty Reeb graph is  0-safe.
For example, the graph drawn in \cref{fig:UpAndDownForests} is not $\delta$-tailed because the bottommost up-fork has no down-path of height $\delta$; in addition, the topmost down-fork has no up-path of height $\delta$.

We next have two results, proved in \fullVerRef{sec:properties_of_truncation}, which show how the $\bullet$-tailed and $\bullet$-safe properties are maintained under smoothing and truncating, albeit with modified parameters.

\begin{restatable}{proposition}{longTailsCombined}
\label{prop:longTailsCombined}
        If $(G,f)$ is $t$-tailed, then $S_\e(G,f)$ is $(t+2\e)$-tailed.
        If $(G,f)$ is $s$-safe, then $S_\e(G,f)$ is $(s+\e)$-safe.
        In particular, $S_\e(G,f)$ is always $2\e$-tailed and $\e$-safe.
\end{restatable}

    \begin{restatable}{lemma}{singleEdge}
        \label{obs:singleEdge}
        Fix $0 \leq \tau \leq \e$.
        If $(G,f)$ is $\e$-tailed or safe, then $T^\tau(G,f)$ is $(\e-\tau)$-tailed or safe, respectively.
    \end{restatable}

    \begin{figure}
        \centering
        \includegraphics{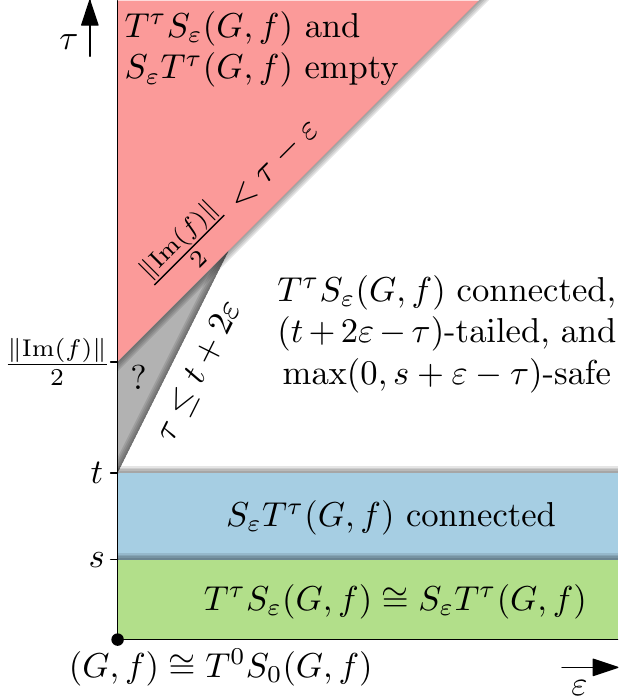}
        \caption{For connected, $t$-tailed, and $s$-safe $(G,f)$, properties of~$S_\e^\tau(G,f) = T^\tau S_\e(G,f)$ and $S_\e T^\tau(G,f)$, parameterized by $\tau$ and $\e$.}
        \label{fig:epstauproperties}
    \end{figure}
Combining \cref{prop:longTailsCombined} and \cref{obs:singleEdge}, we can see that outside the pink and grey regions of \cref{fig:epstauproperties}, we know that $S_\e^\tau(G,f)$ is $(t+2\e-\tau)$-tailed and $(s+\e-\tau)$-safe.

\begin{proposition}
Fix $0 \leq \e$ and $0 \leq \tau$, and assume $(G,f)$ is $t$-tailed and $s$-safe.
If $\tau \leq t + 2\e$ and $\tau \leq \e + \|\Im(G,f)\|/2 $,
then $S_\e^\tau(G,f)$ is $(t+2\e-\tau)$-tailed and $(s+\e-\tau)$-safe.
\end{proposition}
\begin{proof}
Because $(G,f)$ is $t$-tailed, $S_\e(G,f)$ is $(t+2\e)$-tailed by the first statement of \cref{prop:longTailsCombined}.
Since $\tau \leq t+2\e$, $S_\e^\tau(G,f) = T^\tau S_\e(G,f)$ is $(t+2\e-\tau)$-tailed by \cref{obs:singleEdge}.
Similarly, since $(G,f)$ is $s$-safe, $S_\e(G,f)$ is $(t+\e)$-safe by the second statement of \cref{prop:longTailsCombined}.
Then since $\tau \leq t+2\e$, $S_\e^\tau(G,f) = T^\tau S_\e(G,f)$ is $(s+\e-\tau)$-safe by \cref{obs:singleEdge}.
\end{proof}

This brings us to our conclusion of parameters for which the connectivity is maintained, with full details provided in \fullVerRef{ssec:apdx:Connectedness}.

\begin{restatable}{proposition}{ConnectedSmoothingAndTruncating}
\label{cor:ConnectedSmoothingAndTruncating}
If $(G,f)$ is connected and $\tau\in[0,2\e]$, then $S_\e^\tau(G,f)$ is also connected.
\end{restatable}

\begin{proof}
[Sketch proof.]
We show in \cref{lem:staysConnected} that for a connected, $t$-tailed graph, $T^t(G,f)$ is connected by ensuring disjointness of the portion of the graph $G$ removed because it is lacking an up-path, and that which is removed because it is lacking a down-path.
The result is then a corollary of \cref{prop:longTailsCombined}.
\end{proof}

\subsubsection{When do \texorpdfstring{$S^\e$}{Se} and \texorpdfstring{$T^\tau$}{Tt} commute?}

We finally investigate the commutativity of smoothing and truncating.
The example of \cref{fig:smoothingOperation_withTrunc} shows why we must be careful with order of operations since $T^\tau S_\e(G,f)$ is not necessarily the same as $S_\e T^\tau(G,f)$.
Specifically, $S_{2\e}^{2\e}(G,f) = T^{2\e}S_{2\e}(G,f)$ has one connected component, but any smoothing of $T^{2\e}(G,f)$ has two connected components.
However, the next two results imply that this issue does not arise if we smooth sufficiently before truncating.%
\begin{restatable}{proposition}{conflow}\label{prop:conflow01}
    If $(G,f)$ is $\tau$-safe, then $S_\e T^\tau(G,f)\cong T^\tau S_\e(G,f)$.
\end{restatable}
    The proof is provided in \fullVerRef{ssec:apdx:commutes}.
    Combining the proposition with \cref{obs:singleEdge} and \cref{prop:longTailsCombined} gives the surprising result that the functors $T$ and $S$ do commute in the green region of \cref{fig:epstauproperties}.
    We can next use this result to show that for certain choices of $\e$ and $\tau$, we can additively combine the parameters for truncated smoothing.

    \begin{theorem}
    \label{thm:additive}
        If (1) $(G,f)$ is empty or (2) $\tau_1\leq 2\e_1$ and $(G,f)$ is weakly $(\tau_1-\e_1)$-safe,
        then
        $
        S_{\e_2}^{\tau_2} S_{\e_1}^{\tau_1}(G,f)
        \cong
        S_{\e_1+\e_2}^{\tau_1+\tau_2}(G,f)$.
        
    \end{theorem}
    \begin{proof}
        Both smoothing and truncating the empty Reeb graph yields the empty Reeb graph.
        So we are done if $(G,f)$ is the empty Reeb graph, and we obtain not only an isomorphism but an equality.
        Now suppose that $(G,f)$ is not empty.
        Then $S_{\e_1}(G,f)$ is $2\e_1$-tailed and weakly $(\tau_1-\e_1+\e_1)$-safe, and by definition $\min(2\e_1,\tau_1)\geq\tau_1$-safe.
        Therefore $S_{\e_2}T^{\tau_1}S_{\e_1}(G,f)\cong T^{\tau_1}S_{\e_2}S_{\e_1}(G,f)$, and hence using \cref{prop:conflow01},
        \[
        S_{\e_2}^{\tau_2} S_{\e_1}^{\tau_1}(G,f)
        =     T^{\tau_2} S_{\e_2}T^{\tau_1}S_{\e_1}(G,f)
        \cong T^{\tau_2}T^{\tau_1}S_{\e_2}S_{\e_1}(G,f)
        \cong S_{\e_1+\e_2}^{\tau_1+\tau_2}(G,f).
        \qedhere
        \]
    \end{proof}

    In particular, the assumptions of the theorem are satisfied if $\tau_1 \leq \e_1$ since every non-empty graph is $0$-safe.

\subsection{Truncated interleaving distance}
\label{ssec:truncInterleavingSurvey}

In this section, we survey the results related to defining the family of truncated  interleaving distances, proving that certain linear subspaces of our two parameter functor space (shown in \cref{fig:epstauproperties})  form a categorical flow.
Since any category with a flow gives an interleaving distance, we then use truncated smoothing to build a new family of metrics for Reeb graphs.

The whole idea behind building a category with a flow is that the flow itself must be functorial, which means we must have knowledge of how it acts both on objects and morphisms.
So far, the results discussed in \cref{ssec:PropertiesOfTruncationSurvey} only correspond to the object information.
In \cref{sec:mapsAndProperties}, we will describe how to explicitly build the morphisms $S_\e^\tau(G,f) \to S_{\e'}^{\tau'}(G,f)$  (i.e., function preserving maps).
However, these morphisms are only available for certain choices of parameters. 
Restricting our view only to $(\e,\tau)$ pairs for which these morphisms exist gives us that for any choice of $m \in [0,1]$ we can set $\tau = m\e$ to get a flow.

\begin{restatable}{theorem}{CategoricalFlow}
\label{thm:CategoricalFlow}
    For any $m \in [0,1]$,  the map
    $S^m\from  ([0,\infty),\le) \to \mathbf{End}(\Reeb); \e \mapsto S_\e^{m\e}$ is a functor and defines a categorical flow on $\Reeb$.

\end{restatable}
Essentially, this $m$ can be thought of as defining the slope of a line based at the origin in the parameter space of \cref{fig:epstauproperties}, and thus using \cref{thm:flowGivesMetric}, we have an interleaving distance for any line with slope less than 1.

\begin{restatable}{corollary}{truncatedInterleaving}
\label{cor:truncatedInterleaving}
For any $m \in [0,1]$, $S^m$ gives rise to an interleaving-type distance
\[
d_I^m((G,f),(H,h)) := \inf \{ \e \geq 0 \mid \text{there exists a $\e$-interleaving with respect to }S^m \}.
\]
Specifically, $d_I^m$ is an extended pseudo-metric.
\end{restatable}

In the next theorem, we show that with the exception of $m=1$, all the metrics created are closely related in the following sense.
Two metrics $d_A$ and $d_B$ are said to be \emph{strongly equivalent} if there are positive constants $\alpha_1$ and $\alpha_2$ such that $\alpha _1 d_A \leq d_B \leq \alpha_2 d_A$.
In the following theorem, we show that $d_I^m$ and $d_I^{m'}$ are strongly equivalent if  $(m,m')$ is contained in the white region of \cref{fig:m_vs_m-prime}.

\begin{restatable}{theorem}{EquivalentMetrics}
\label{thm:EquivalentMetrics}
    For any pair 
    $m,m' \in [0,1)$ with $0 \leq m'-m < 1-m'$
    the metrics $d_I^m$ and $d_I^{m'}$ are strongly equivalent.
    Specifically, given Reeb graphs $(G,f)$ and $(H,h)$,
    \begin{equation*}
    d_I^m((G,f), (H,h))\leq d_I^{m'}((G,f), (H,h))\leq \frac{1-m}{1-m'}d_I^m((G,f), (H,h))
    \end{equation*}
\end{restatable}

The proof of this theorem is contained in \fullVerRef{sec:appendix_strongequivmetrics}.
Of course, as long as we are willing to loosen the bounds, this result extends to any pair of $m,m' \in [0,1)$.

\begin{figure}
    \centering
    \includegraphics{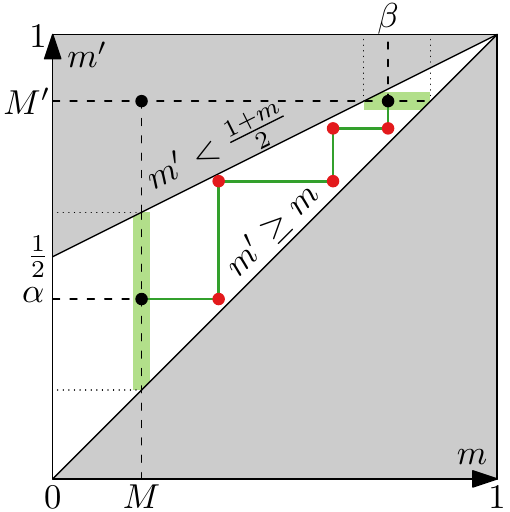}
    \caption{Parameter space for comparing metrics $d_I^{m}$ and $d_I^{m'}$.
    The white region is allowable pairs for \cref{thm:EquivalentMetrics}.
    The vertices of the zigzag (shown as red points) give pairs of strongly equivalent metrics which, when combined, show that $M$ and $M'$ are strongly equivalent in \cref{cor:equivalenceAll_m}.}
    \label{fig:m_vs_m-prime}
\end{figure}
\begin{corollary}
\label{cor:equivalenceAll_m}
For all pairs $0 \leq M \leq M' <1$,
there exist positive constants $C_1$ and $C_2$ dependent on $M$ and $M'$ such that
\begin{equation*}
    C_1 d_I^M((G,f),(H,h))
    \leq d_I^{M'}((G,f),(H,h))
    \leq C_2 d_I^M((G,f),(H,h)),
\end{equation*}
and thus $d_I^M$ and $d_I^{M'}$ are strongly equivalent metrics.
\end{corollary}

\noindent\textit{Proof.}
Consider $M$, $M'$ given with  $M \leq M'$.
If $M' \leq \frac{1+M}{2}$, then \cref{thm:EquivalentMetrics} applies directly.  Otherwise, we assume  that $M' \geq \frac{1+M}{2}$.
Then $d_I^M$ is equivalent to $d_I^{\alpha}$ for any $\alpha$ in the interval 
$(M,\frac{1+M}{2})$ and $d_I^{M'}$ is equivalent to $d_I^{\beta}$ for any $\beta$ in the interval $(2M'-1,M')$.
Then there is a zigzag like the example in \cref{fig:m_vs_m-prime} between $\alpha$ and $\beta$ which remains in the white region and for which each adjacent pair are strongly  equivalent metrics.
Equivalence of metrics is transitive, so this implies $d_I^M$ and $d_I^{M'}$ are equivalent.
\qed

In particular, this corollary gives that the original Reeb graph interleaving distance (where $m=0$) is strongly equivalent to $d_I^m$ for all $m \in[0,1)$.
We note that there are many possible zigzag paths which can be used to obtain this bound, but further exploration is needed to determine which, if any, provide optimal constants.

\subsection{Properties of the metrics}
\label{ss:propertiesOfMetric}

As noted, $d_I^m$ is an extended pseudometric, which means that it is possible for $d_I^m((G,f), (H,h))$ to be infinite.
However, it turns out this is not the case for broad classes of graphs.  In fact, in order to take infinite value, there must be no $\e$-interleaving with respect to $S^m$ between the two Reeb graphs.
That being said, there are very specific instances where this metric takes on infinite value.

The easiest case to handle is when $m \in [0,1)$, since we can use the characterization given in \cite{deSilva2016} in conjunction with the equivalence of metrics \cref{cor:equivalenceAll_m}.

\begin{proposition}
Let $m \in [0,1)$.
Then $d_I^m((G,f),(H,h)) < \infty $ iff $G$ and $H$ have the same number of path-connected components.
\end{proposition}
\begin{proof}
Note that $d_I^0 = d_I$. 
By \cite[Prop.~4.5]{deSilva2016}, $d_I((G,f), (H,h))$ is finite if and only if $G$ and $H$ have the same number of path connected components.
This combined with \cref{cor:equivalenceAll_m} gives the proposition.
\end{proof}

The characterization of when $d_I^m$ is infinite for $m=1$ is more complicated.
Consider a connected graph $(G,f)$ with $\Im(G,f) = [a,b]$.
When $m=1$, we are interested in understanding the behavior of $S_\e^\e(G,f)$.
By \cref{prop:ImageOfSmoothingAndTruncating}, we see that $b-a \geq 2(\tau-\e) = 0$, so  $S_\e^\e(G,f) = [a,b]$.
That is to say that the image of $(G,f)$ is unchanged by $S_\e^\e$.
Now, if we wanted to determine the interleaving distance $d_I^1$ for a given $(G,f)$ and $(H,h)$, one requirement is always that we must smooth the given graphs enough for there to be a morphism $(G,f) \to S_\e^\e(H,h)$.
However, because $S_\e^\e$ does not change the image, the function preserving requirement of morphisms mean that if the graphs did not start with the same image no choice of $\e$ will make this possible.
With this example in mind, we can characterize when $d_I^m$ takes on infinite values for $m=1$.

\begin{proposition}
\label{prop:m1_infinite}
Let $m =1$ and assume $G$ and $H$ are connected.
Then $d_I^m((G,f),(H,h)) < \infty$ if and only if $\Im(G,f) = \Im(H,h)$.

Further, if $\Im(G,f) = \Im(H,h)$, then  $d_I^m((G,f),(H,h)) \leq |\Im(G,f)|$. 

\end{proposition}

\begin{proof}
Note that by \cref{prop:ImageOfSmoothingAndTruncating}, for any connected $G'$ with $\Im(G',f') = [a,b]$, $\Im S_\e^\e(G',f') = [a,b]$.
So the truncated smoothing maintains the image for every connected component, and thus for the union of the connected components.
Thus, we have $\Im(G,f) = \Im(S_\e^\e(G,f))$ and $\Im(H,h) = \Im(S_\e^\e(H,h))$ for any choice of $\e$.

Assume we have an $S_\e^\e$ interleaving $\phi\from (G,f) \to S_\e^\e(H,h)$ and $\psi\from (H,h) \to S_\e^\e(G,f)$.
Because $\phi$ and $\psi$ are  function preserving,
$\phi(G) = \Im(G,f) \subseteq \Im(S_\e^\e(H,h))$
and
$\psi(H) = \Im(H,h) \subseteq \Im(S_\e^\e(G,f))$.
But since $S_\e^\e$ leaves the images unchanged, this implies that $\Im(G,f) = \Im(H,h)$.

Now assume $\Im(G,f) = \Im(H,h)$.  
Let $\e = |\Im(G,f)|$ and consider the thickening $G \times [-\e,\e]$ and a value $a \in \Im(G,f)$. 
We claim that $(f+\Id)\inv(a)\subseteq G \times [-\e,\e]$ is exactly $A = \{(x,a-f(x))\mid x\in G\}$ and in particular, that it is homeomorphic to $G$.
Indeed, for any $x \in G$, $a-f(x)\in [-\e,\e]$ and the point $y = (x,a-f(x))$ has image $(f+\Id)(y) = a$ so $A \subseteq (f+\Id)\inv(a)$. 
Moreover, for any $(x,t) \in (f+\Id)\inv(a)$, $f(x) + t = a$ so $t = a-f(x)$, thus $(f+\Id)\inv(a) \subseteq A$. 

So, since $G$ is connected and $f$ is continuous, $(f+\Id)\inv(a) \cong G$ is a single connected component for any $a \in \Im(G,f)$, and the same is true for $(H,h)$.
Because the $S_\e^\e$ smoothing maintains the image, this implies  $S_\e^\e(G,f) = S_\e^\e(H,f)$ is a single line segment with the same image. 
We obtain an interleaving by simply sending  every point in $(G,f)$ to the unique point at the same height in $S_\e(H,h)$ and vice versa, so the $d_I^m$ distance is finite. 
\end{proof}

We next investigate stability, for this collection of metrics.
\begin{definition}
Let $(\X,f)$ and $(\X,g)$ be $\R$-spaces with the same total space $\X$, and let $R(\X,f)$ and $R(\Y,g)$ be the respective Reeb graphs. A metric $d$ is said to be stable if
\begin{equation*}
    d(R(\X,f),R(\X,g)) \leq \|f-g\|_\infty.
\end{equation*}
\end{definition}
The original Reeb interleaving distance, $m=0$, is stable \cite[Thm 4.4]{deSilva2016}.
Unfortunately, $d_I^m$ is not stable in the strictest sense; to see why, consider the following simple example.
Consider two simple line segments for graphs, for example, $(L,f_1)$ and $(L,f_2)$ where
$\Im(L,f_1) = [-a,a]$
and
$\Im(L,f_2) = [-b,b]$ for $a < b$.
Then $\|f_1-f_2\|_\infty = b-a$.
However, the interleaving distance requires that we smooth at least until $[-b,b] = \Im(L,f_2) \subseteq \Im(S_\e(L,f_1))$.
But by \cref{prop:ImageOfSmoothingAndTruncating}, $\Im(S_\e(L,f_1)) = [a-(\e-m\e),a+(\e-m\e)]$.
Thus $d_I^m(f_1,f_2) \geq \frac{b-a}{1-m} \geq b-a$, and is strictly greater if $m \neq 0$.
This means that $b-a = \|f_1-f_2\|_\infty < d_I^m((L,f_1),(L,f_2))$, and thus $d_I^m$ is not stable.

We can regain at least partial control of the distance, however, as $d_I^m$ is still Lipschitz when given a fixed choice of $m$.

\begin{proposition}
\label{prop:stabilityWIthConstant}
Let $m  \in [0,1)$.
Assume $(\X,g_1)$ and $(\X,g_2)$ are given for a connected space $\X$ and denote the associated Reeb graphs by $(G,f)$ and $(H,h)$ respectively.
Then there is a positive constant $C$ dependent on $m$ for which
\begin{equation*}
    d_I^m((G,f), (H,h)) \leq C \|g_1 - g_2\|_\infty.
\end{equation*}
\end{proposition}

\begin{proof}
By \cref{cor:equivalenceAll_m}, $d_I^0$ and $d_I^m$ are strongly equivalent metrics, so there is a positive constant $C$ for which $d_I^m \leq C d_I^0$.
Then because the Reeb graph interleaving distance $d_I^0$ is stable, we have
\[d_I^m((G,f), (H,h)) \leq C d_I^0((G,f), (H,h)) \leq C \|g_1-g_2 \|.\qedhere\]
\end{proof}

Because of the dependence on \cref{cor:equivalenceAll_m} where the optimal choice of zigzag to find the constant $C$ is unclear, we do not give an explicit formulation here.
We conclude by connecting our extended pseudometric to two other metrics for Reeb graphs, the functional distortion distance \cite{Bauer2014} and the bottleneck distance \cite{Oudot2017a}.
The proof is a straightforward implication of inequalities, so due to space constraints we simply state these results without formally defining either. 
The interested reader  can find further details on the metrics in \cite{Bauer2014} and \cite{Oudot2017a}.

\begin{proposition}
\label{Prop:FD}
The truncated interleaving distance is strongly equivalent to the functional distortion distance. 
Further, defining $d_B$ as the bottleneck distance of the level set persistent homology, we have the inequality
        $d_B  \leq 5 d_I^m$.
\end{proposition}
\begin{proof}
The interleaving distance, $d_I^0$, is strongly equivalent to
the functional distortion distance  by \cite[Thm 16]{Bauer2015b}.
So by \cref{cor:equivalenceAll_m} and transitivity of strong equivalence, they are each strongly equivalent to $d_I^m$ for any $m \in [0,1)$.
To obtain the inequality, we use the bound on the bottleneck distance of level set persistent homology by the Reeb graph interleaving distance in  \cite[Thm.~4.13]{Botnan2018a}.
\end{proof}

\section{Categories and interleavings }
\label{sec:CategoriesBackground}

 A key tool in our operations on Reeb graphs comes from a category theory perspective, so we begin filling in the holes left behind in \cref{ssec:truncInterleavingSurvey} by not fully describing the categorical aspects of the  constructions discussed.
 We briefly review some essential concepts, but refer the reader to \cite{Riehl2017} for a background in category theory, as well as to prior work on category theory for Reeb graphs~\cite{Munch2016} and categories with a flow \cite{deSilva2018} for more details.

 \subsection{Categories, flows, and interleaving distances}
 \label{ssec:categories}
 A category $\CC$ is a collection of objects, a collection of morphisms between the objects.
 We further require an associative composition operator which can compose any two of the morphisms, and that every object $c \in \CC$ has an identity morphism $\Id_c\from c \to c$.
 Mathematics is full of examples, from sets to vector spaces, as well as constructible $\mathbb{R}$-spaces.
 Denote by {\Reeb} the  category of Reeb graphs with  function preserving maps as the morphisms.

 A functor $F\from \CC \to \DD$ is a map between any two categories sending objects to objects: $F(c) = d$;
 and morphisms to morphisms: $F[\phi]\from F(c) \to F(d)$ for $\phi\from c \to d$.
 This collection of data must preserve composition and identities, so $F[\phi \circ \psi] = F[\phi]\circ F[\psi]$ and $F[\Id_c] = \Id_{F(c)}$.
 Some examples of useful functors are homology $H_k\from\mathrm{Top} \to \mathrm{Vect}$ from topological spaces to vector spaces (assuming field coefficients), or the functor $\pi_0\from \mathrm{Top} \to \mathrm{Set}$ sending a topological space to the set of its path-connected components.

 We can treat the collection of functors from $\CC$ to $\DD$ as a category in itself, where the morphisms are given by \textit{natural transformations}.
 Specifically, given $F,G\from \CC \to \DD$, a \textit{natural transformation} $\eta\from F \Rightarrow G$ is a collection of morphisms $\eta_c\from F(c) \to G(c)$ such that
 \begin{equation*}
     \begin{tikzcd}
        F(c) \ar[r,  "\eta_c"] \ar[d, "{F[\phi]}"'] & G(c) \ar[d, "{G[\phi]}"] \\
        F(c') \ar[r, "\eta_{c'}"] & G(c')
     \end{tikzcd}
 \end{equation*}
 commutes for any morphism $\phi:c \to c'$ in $\CC$.
 This functor category  with objects as functors and morphisms given by natural transformations is denoted $\DD^\CC$.
 A natural transformation is a \textit{natural isomorphism} if every map $\eta_c$ is an isomorphism.
 When we have a natural isomorphism between functors we write $F \cong G$.
  A special case of the functor category is when $\DD = \CC$.
 A functor from a category to itself, $F:\CC \to \CC$, is called an endomorphism and the category of all such functors with natural transformations is denoted $\End(\CC)$.

The end goal of this paper is to  study flows of Reeb graphs, where the idea is to have a 1-parameter varying collection of Reeb graphs satisfying nice properties.
For this, we look to the definition of a category with a flow given in \cite{deSilva2018}.

 \begin{definition}
 \label{defn:Flow}
  Let $[0,\infty)$ denote the poset category of positive real numbers with morphisms given by $\leq$.
 Given a category $\CC$, a \emph{categorical flow} is a functor $F:[0,\infty) \to \End(\CC)$, $\e \mapsto F_\e$, with $F_0 \cong \1_\CC$ and $F_aF_b \cong F_{a+b}$ for all $a,b \geq 0$.
 \end{definition}

 Note that this definition is hiding quite a bit of infrastructure.
 In particular, $F_\e$ is a functor, so it gives rise to a morphism $F_\e[\phi]$ for every morphism $\phi$ of $\CC$.
 The fact that $F$ is a functor means we get a natural transformation $F[\e \leq \e']: F_\e \Rightarrow F_{\e'}$ for every $\e \leq \e'$.
 That this is a natural transformation means that
 \begin{equation*}
 \begin{tikzcd}[column sep = huge]
    F_\e(c)
        \ar[r,  "{F[\e \leq \e']_c}"]
        \ar[d, "{F_\e[\phi]}"']
    & F_{\e'}(c)
        \ar[d, "{F_{\e'}[\phi]}"]
    \\
    F_{\e}(c')
        \ar[r, "{F[\e \leq \e']_{c'}}"']
    & F_{\e'}(c')
 \end{tikzcd}
 \end{equation*}
 commutes for every morphism $\phi$ of $\CC$.
 The final requirement, $F_aF_b \cong F_{a+b}$,  checks that flowing by $a$ and then $b$ is at least closely related to flowing by the total amount $a+b$ all at once.

This definition is particularly useful since it can be used to provide an interleaving distance for any category with a given flow.

\begin{definition}
\label{defn:interleavingGeneral}
Given a category with a (categorical) flow $(\CC,F)$ and two objects $X,Y \in \CC$, an $\e$-interleaving of $X$ and $Y$ is a pair of morphisms ${\color{violet}\phi}: X \to F_\e Y$ and ${\color{OliveGreen}\psi}: Y \to F_\e X$ such that
\begin{equation}
\label{eq:categoricalInterleaving}
    \begin{tikzcd}[column sep = 10em]
        X \ar[r, "{F[0 \leq \e]}"]
            \ar[dr, violet, "\phi" description, near start]
        & F_\e X
            \ar[dr, violet, "{F_\e[\phi]}" description, very near start]
            \ar[r, "{F[\e \leq 2\e]}"]
        & F_{2\e} X \\
        Y
            \ar[ur, OliveGreen, dashed, crossing over, "\psi" description,  near start]
            \ar[r, "{F_\e[0 \leq \e]}"']
        & F_\e Y
            \ar[ur,  OliveGreen, dashed, crossing over, "{F_\e[\psi]}" description, very near start]
            \ar[r, "{F[\e \leq 2\e]}"']
        & F_{2\e Y}
    \end{tikzcd}
\end{equation}
commutes.
Then, the interleaving distance is given by
\begin{equation*}
   d_{(\CC,F)}(X,Y) = \inf \{ \e \geq 0 \mid X,Y \text{ are } \e\text{-interleaved} \}.
\end{equation*}
\end{definition}
Note that $d_{\CC,F}$ is an extended pseudometric on the objects of $\CC$ \cite[Thm.~2.7]{deSilva2018}; i.e.~it can take infinite value, and $d(X,Y) = 0$ does not imply that $X = Y$.

\begin{remark}
\label{remark:abusingNotation}
It is necessary to now point out that we are consciously abusing notation from here on out.
This categorical flow is a special case of the definition of flow given in \cite{deSilva2018,Stefanou2018}; specifically, what we have defined is called a strong flow in that work.
In \cite{deSilva2018}, the flow comes with additional notation to encode the isomorphisms of  $T_0 \cong \1_\CC$ and $T_aT_b \cong T_{a+b}$, and to ensure that they interact accordingly.
Then, when giving the interleaving definition, the interleaving diagram is expanded to essentially be comprised of linked pentagons rather than the large scale triangles seen in \cref{eq:categoricalInterleaving}.
In this paper, we will do our best to point out when it happens, but we will suppress the isomorphism since it does not serve to illuminate the work, but rather invariably results in exponential growth of the size of the required commutative diagrams.
\end{remark}

\subsection{Smoothing and the interleaving distance for Reeb graphs}
\label{ssec:Reebinterleaving}
In this section, we give further specifics of the original smoothing definition from \cite{deSilva2016}, given as \cref{defn:originalSmoothing}.
While the idea comes from the equivalence of categories between $\Reeb$ and a particular category of cosheaves, we will not need that construction here so we will focus on the geometric definition of smoothing.

We construct the thickening of the graph, $(G \times [-\e,\e], f+\Id)$, and the smoothing $S_\e(G,f)$ is the Reeb quotient of this space, where we denote the quotient map by $q$.
However, this leaves out the important collection of morphisms that come with the smoothing construction.
Specifically, we have an inclusion $(\Id,  0): G \to G \times [-\e,\e]$ given by $x \mapsto (x,0)$; see \cref{fig:thickening} for an illustration, where this inclusion is shown in the thickened space as a dotted copy of the original graph from the left.
Let $\eta=q_\e\after(\1 \times 0)\from G\to G_\e$, so that $\eta(x)=q_\e(x,0)$.
The process can be summarized in the diagram
  \begin{equation}
  \label{eq:smoothingTriangle}
  \begin{tikzcd}
    & (G \times [-\e,\e], f+\Id) \ar[dr, "q"]\\
    (G,f) \ar[ur,hook,"{(\mathrm{Id},0)}"] \ar[rr, "\eta"] & &
    S_\e(G,f).
  \end{tikzcd}
  \end{equation}
Note that $\eta$, $q$, and $(\Id,  0)$ are all function preserving maps.

Given a morphism $\phi:(G,f) \to (H,h)$, i.e.~a function preserving map $\phi:G \to H$, it can be checked that there is an induced morphism $S_\e[\phi]:S_\e(G,f) \to S_\e(H,h)$
making the diagram
    \begin{equation}
    \label{eq:S_e_appliedtomaps}
  \begin{tikzcd}
    & (G \times [-\e,\e], f+\Id)
        \ar[dr, "q"]
        \ar[dd, violet,"{(\phi,\Id)}", near start]
    \\
    (G,f)
        \ar[ur,hook,"{(\mathrm{Id},0)}"]
        \ar[rr, "\eta", near start, crossing over]
        \ar[dd,violet, "\phi" ]
    & &
    S_\e(G,f)=: (G_\e,f_\e) \ar[dd, violet, "{S_\e[\phi]}"]
    \\
    &
    (H \times [-\e,\e], h+\Id)
        \ar[dr, "q"]
    \\
    (H,h)
        \ar[ur,hook,"{(\mathrm{Id},0)}"]
        \ar[rr, "\eta"] & &
    S_\e(H,h)=: (H_\e,h_\e)
  \end{tikzcd}
  \end{equation}
  commute.
With these $S_\e[\phi]$ maps, $S_\e$ is an endofunctor on $\Reeb$; see \cite[Sec 4.4]{deSilva2016} for details.

Further, note that replacing $(H,f)$ with $S_\e(G,f)$ in \cref{eq:S_e_appliedtomaps} gives a map $S_\e[\eta]:S_\e(G,f) \to S_\e S_\delta(G,f)$.
As noted in \cite[Obs.~4.30]{deSilva2016}, there is a natural isomorphism $S_{\e_2} S_{\e_1}(G,f) \to S_{\e_1+\e_2}(G,f)$.
For this reason and in the spirit of \cref{remark:abusingNotation}, we abuse notation and write $\eta: S_{\e_1}(G,f) \to S_{\e_1+\e_2}(G,f)$ for the composition of maps
  \begin{equation*}
      \begin{tikzcd}
        S_{\e_1}(G,f) \ar[d, "{S_\e[\eta]}"']
          \ar[dr, "\eta"]\\
        S_{\e_2}(S_{\e_1}(G,f)) \ar[r,"\cong"] & S_{\e_1+\e_2}(G,f).
      \end{tikzcd}
  \end{equation*}
In particular, for any $0 \leq \e \leq \e'$, we write $\eta:S_\e(G,f) \to S_{\e'}(G,f)$ without changing the notation $\eta$ unless it is necessary for clarity.
A full discussion of this suppressed homeomorphism  is in \cref{ssec:apdx:isomorphism}.

With this notation, we have that the $\eta$'s compose, in the sense that
\begin{equation*}
    \begin{tikzcd}
        S_{\e_1}(G,f)
            \ar[r, "\eta"]
            \ar[rr, bend right, "\eta"]
        & S_{\e_2}(G,f)
            \ar[r, "\eta"]
        & S_{\e_3}(G,f)
    \end{tikzcd}
\end{equation*}
commutes for all $\e_1 \leq \e_2 \leq \e_3$.
Again suppressing the homeomorphism, the $\eta$ maps also interact with the Reeb quotient map in the sense that
\begin{equation}
\label{eq:Square}
    \begin{tikzcd}
        G \times [-\e,\e] \ar[r, hook] \ar[d,"q"']
            & G \times [-\e', \e']\ar[d,"q"]\\
        S_{\e}(G,f) \ar[r,"\eta"] & S_{\e'}(G,f)
    \end{tikzcd}
\end{equation}
commutes for all $\e \leq \e'$.
Another useful diagram to note is that with this supression of the homeomorphism, the diagram
\begin{equation}
\label{eq:Se_of_morphism_and_eta_square}
\begin{tikzcd}
    S_\e(G,f) \ar[r, "\eta"]
    \ar[d, "{S_\e[\psi]}"']
    & S_{\e'}(G,f)
    \ar[d, "{S_{\e'}[\psi]}"]
    \\
    S_\e(H,h) \ar[r, "\eta"]
    & S_{\e'}(H,h)
\end{tikzcd}
\end{equation}
commutes.

All of this bookkeeping can be summarized by $S_0(G,f) \cong (G,f)$ and $S_b S_a(G,f) \cong S_{a+b}(G,f)$, and thus $S: \e \mapsto S_\e$ with $S_\e \in \End(\Reeb)$ is a strong flow.
 Finally, since $S$ defines a flow, we have an interleaving distance (\cref{defn:interleavingGeneral}) on $\Reeb$ which we will call simply the interleaving distance, $d_I = d_{\Reeb,S}$.

 \begin{proposition}
[{\cite[Props.~4.3, 4.5, and 4.6]{deSilva2016}} ]
 The Reeb interleaving distance $d_I$ is an extended pseudometric.
 It can take infinite value if and only if the Reeb graphs have different numbers of path components.
 It is 0 if and only if the Reeb graphs are isomorphic (i.e.~if their graphs are isomorphic and that isomorphism is function preserving).

 \end{proposition}

\section{Equivalent definitions of truncated smoothing }
\label{sec:EquivalentDefinitions}
We have defined truncated smoothing by removing the portion of $S_\e(G,f)$ corresponding to growing tails (\cref{defn:TruncatedSmoothing}).
In this section, we show three other equivalent definitions (\cref{prop:defn_combinatorial,cor:DefnIntersection,prop:DefnBackwardView}) which  will prove useful later when for calculations or proving mathematical properties.

\subsection{Definition in terms of quotient maps}
\label{sec:conSmoothing}

We begin by presenting a definition of truncated smoothing that is more closely tied to the thickening definition used to construct $S_\e(G,f)$ in the first place by restricting the values of $\tau$ that can be used.
Given $\e \geq 0$ and $\tau \in [0,2\e]$,
we have the commutative diagram
\begin{equation*}
    \begin{tikzcd}
    (G \times [\tau - \e, \e], f+\Id)
        \ar[d,hook]  \ar[dr,dashed, "q"] \\
    (G \times [-\e, \e], f+\Id)
        \ar[r,"q"] & S_\e(G,f)  \\
    (G \times [-\e, \e-\tau], f+\Id)
        \ar[u,hook']  \ar[ur,dashed, "q"']
    \end{tikzcd}
\end{equation*}
where the middle map $q$ is the Reeb quotient map, and the diagonal arrows are its restrictions.
We can then study the intersection
\begin{equation}
\label{eq:TrunSmoothingDefn}
    q\Big(G \times [\tau-\e,\e]\Big)
    \,\cap \,
    q\Big(G \times [-\e,\e-\tau]\Big)
\end{equation}
with function given by the restriction of the function from $S_\e(G,f)$.
See \cref{fig:truncated_smoothing_construction} for an example with two different choices of $\tau$ relative to $\e$.
The following proposition shows that this intersection gives an equivalent definition for $S_\e^\tau$ to that of \cref{defn:TruncatedSmoothing}.

\begin{figure}
    \centering
    \includegraphics{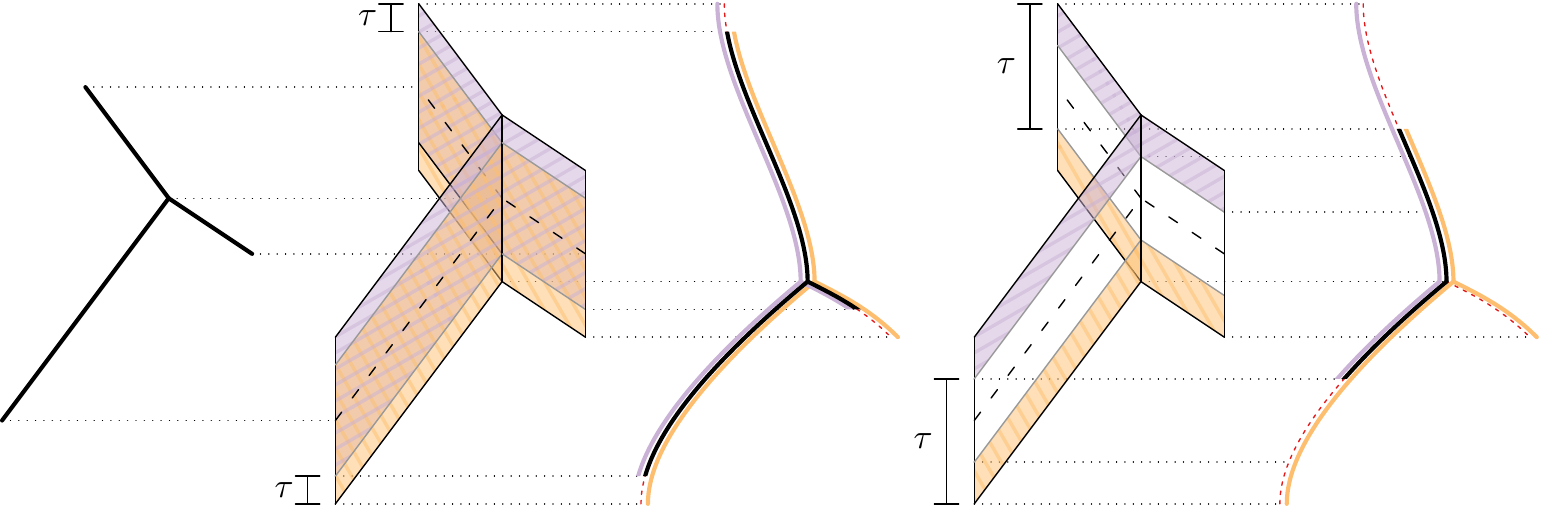}
    \caption{Left: a Reeb graph $(G,f)$. Middle: its thickening and truncated smoothing $T^\tau S_\e(G,f)$ for $\tau \in [0,\e]$. Right: the same for $\tau \in (\e, 2\e]$.}
    \label{fig:truncated_smoothing_construction}
\end{figure}

\begin{proposition}
    \label{prop:defn_combinatorial}
    For $\tau\in[0,2\e]$,
    \begin{equation*}
    S_\e^\tau(G,f) =
        q\Big(G \times [\tau-\e,\e]\Big)
        \,\cap \,
        q\Big(G \times [-\e,\e-\tau]\Big).
    \end{equation*}
\end{proposition}
\noindent \textit{Proof. }
For the sake of notation, denote
$A = q\Big(G \times [\tau-\e,\e]\Big)
        \,\cap \,
        q\Big(G \times [-\e,\e-\tau]\Big)$ for the remainder of the proof.
    We first show that any point
    $x\in A$
    has both an up- and a down-path of height $\tau$ in $S_\e(G,f)$, and hence  $x \in T^\tau S_\e(G,f) = S_\e^\tau(G,f) $.
    Let $(y,t)\in G\times[\tau-\e,\e]$ and
    $(y',t')\in G\times[-\e,\e-\tau]$ such that
    $x=q(y,t)=q(y',t')$.
    Then we have paths $\gamma,\gamma'\from[0,\tau] \rightsquigarrow S_\e(G)$
    given by
    $\gamma(s) = (y,t-s)$ and $\gamma'(s) = (y',t'+s)$.
    Therefore $q\circ\gamma$ is a down-path and $q\circ\gamma'$ is an up-path of $x$ in $S_\e(G,f)$ of height $\tau$, so $S^\tau_\e(G,f) \subseteq T^\tau S_\e(G,f)$.

    For the other direction, we show that any point $x\in S_\e(G,f)\setminus A$ has no up-path or no down-path of height at least $\tau$ in at least one direction.
    Since $x\notin A$, we have
    $q^{-1}(x)\cap(G\times[\tau-\e,\e])=\emptyset$ or
    $q^{-1}(x)\cap(G\times[-\e,\e-\tau])=\emptyset$ by definition.
    Without loss of generality, assume that
    $q^{-1}(x)\cap(G\times[-\e,\e-\tau])=\emptyset$
    as the other case is symmetric.

    First, consider the superlevelset $\{p\in G\times[-\e,\e]\mid f(p)\geq f(x)\}$ and let $C$ be its component containing $q^{-1}(x)$.
    For a point $p=(y,\lambda)\in C$, we define $\beta(p)$ to be the point $(y,\lambda')\in C$ that minimizes $\lambda'$.
    Specifically, we set $\beta(p)=(y,\max\{f(x)-f(y),-\e\})$, so that we project the path down as far as possible towards $q^{-1}(x)$ in its vertical component of the superlevel set.
    Note that $f(\beta(p))=\max\{f(x),f(y)-\e\}$; see \cref{fig:annoyingProof} for a visual of the notation.

    \begin{figure}
        \centering
        \includegraphics[width = 2in]{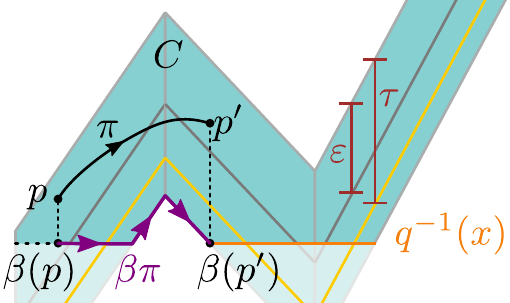}
        \caption{The notation used in the proof of \cref{prop:defn_combinatorial}.}
        \label{fig:annoyingProof}
     \end{figure}
    We claim that for every $p \in C$, we have $f(\beta(p))=f(x)$, which will later imply that no up-path of sufficient height exists.
    Seeking a contradiction, suppose instead that there is a $p\in C$ with $f(\beta(p))>f(x)$.
    This implies that $f(\beta(p)) = f(y)-\e > f(x)$.
    Clearly $\beta(p) \not \in q\inv(x)$, otherwise $f(\beta(p)) = f(x)$.
    Because $C$ is path-connected, there exists a path
    $\pi\from[0,1]\to C$
    from $p$ to a point $p' \in C$ with $\beta(p')\in q^{-1}(x)$; again, see \cref{fig:annoyingProof}.
    Denote $\pi(t)= (y(t),\lambda(t))$.
    Without loss of generality, assume that $\beta(\pi(t))\notin q^{-1}(x)$ for all $t<1$.
    If $f(\beta \pi(t)) = f(x)$ for all $t$, then $\beta(p)$ is in the same level-set connected component as $q\inv(x)$, contradicting that $\beta(p) \not \in q\inv(x)$.
    The assumption that $\beta\pi(t) \not \in q\inv(x)$ for $t <1$ thus further implies that there is a $\delta$ for which $f(\beta\pi(t))>f(x)$ for  $ t \in [1-\delta,1)$.
    By definition of $\beta$, this means that $f(\beta\pi(t)) = f(y(t)) - \e$ for $t\in[1-\delta,1)$.
    However, by continuity of $t\mapsto f(y(t))-\e$, we have $f(\beta\pi(1))=f(x)=f(y(1))-\e$.
    Therefore $\beta\pi(1)=(y,-\e)\in G\times[-\e,\e-\tau]$, and $\beta\pi(1)  = \beta(p') \in q\inv(x)$, contradicting that
    $q^{-1}(x)\cap(G\times[-\e,\e-\tau])=\emptyset$.

    Now, for any $p\in C$, we have $\beta(\pi(p))\in q^{-1}(x)$, so the image of $C$ under $f$ is a subset of $[f(x),f(x)+\tau)$.
    The preimage of any up-path in $S_\e(G,f)$ starting at $x$ lies in $C$, so no such path can have height at least $\tau$, completing the proof.
    \qed

\subsection{Alternative definitions for \texorpdfstring{$\tau \in [0,\e]$}{t in [0,e]}}
Note that if $\tau \in (\e, 2\e]$, then $[-\e,\e - \tau] \cap [\tau - \e, \e] $ is empty.
In this case, $G \times [-\e,\e - \tau] \cap G \times [\tau - \e, \e] $ is empty, which is the reason for to passing to the image of $q$ before intersection in \cref{eq:TrunSmoothingDefn}.
However, if $\tau  \in [0,\e]$, we have no such issue, which leads to an alternative definition for $S_\e^\tau(G)$ which is a corollary to the following lemma.

\begin{lemma}
\label{lem:intersections_and_q}
    Given $(G,f)$, if $\tau \in [0,\e]$,
    $$
    q\Big(G \times [-\e,\e - \tau]\Big)\, \cap \, q\Big(G \times [\tau - \e, \e]\Big)
    =
q \Big(G \times [-\e+\tau,\e - \tau] \Big)
    $$
\end{lemma}

\begin{proof}
First, note that
$
q \left(G \times [-\e,\e - \tau] \cap G \times [\tau - \e, \e]\right)
=
q \left(G \times [-\e+\tau,\e - \tau] \right).
$
The left inclusion of the lemma is immediate since given any
$x \in q(G \times [-\e+\tau, \e - \tau])$
there is
$(y,t) \in G \times [-\e + \tau, \e - \tau]$
such that $q(y,t) = x$.
Thus $(y,t)$ is in both $G \times [-\e,\e - \tau]$ and $G \times [\tau - \e, \e]$, and so $x \in q(G \times [-\e,\e - \tau]) \cap q(G \times [-\e+\tau,\e])$.

For the right inclusion, let $x \in q(G \times [-\e,\e - \tau]) \cap q(G \times [\tau - \e, \e])$.
Then there is
$(y,t) \in G \times[-\e,\e - \tau]$ and
$(y',t') \in G \times[-\e + \tau ,\e]$ such that
$x = q(y,t) = q(y',t')$.
If either $t$ or $t'$ are contained in $[-\e + \tau, \e-\tau]$, then we are done, so we can assume $t \in [\e - \tau, \e]$ and $t' \in [-\e, -\e+\tau]$.
Because they both have the same image under $q$, there is a path $(\gamma_1,\gamma_2) = \gamma\from (y,t) \rightsquigarrow (y',t')$ with $f(\gamma_1(s)) + \gamma_2(s)$ constant; specifically, $q(\gamma(s)) = x$ for all $s \in [0,1]$.
As $\gamma_2$ is a continuous map, there must be an $s \in [0,1]$ for which $\gamma_2(s) \in [-\e+\tau, \e-\tau]$.
Then $\gamma(s) \in G \times [-\e+\tau, \e-\tau]$ and $q(\gamma(s)) = x$, completing the proof.
\end{proof}

Combined with \cref{defn:TruncatedSmoothing}, this gives us an immediate corollary that can be viewed as an equivalent definition for the truncated smoothing whenever $\tau$ is small enough.

\begin{corollary}
\label{cor:DefnIntersection}
    Given $(G,f)$ and $\tau \in [0,\e]$,
\[
S_\e^\tau(G,f) = q(G \times [\tau-\e,\e-\tau])
\]
where $q$ is a restriction of the quotient map
$q: G \times [-\e,\e] \rightarrow S_\e(G,f)$,
and the Reeb graph function is given by the restriction of the function from $S_\e(G,f)$.
\end{corollary}

Another way of viewing the truncated definition is by looking backward in the flow by $\tau$.
Because flows are only defined for $\e \geq 0$, we can use this viewpoint only if $\tau$ is small enough that $\e-\tau$ is non-negative, thus we have the following equivalent definition for truncated smoothing with small enough $\tau$.
\begin{corollary}
\label{prop:DefnBackwardView}
    Given $(G,f)$ and $\tau \in [0,\e]$, then
    $$
    S_\e^\tau(G,f) = \eta(S_{\e-\tau}(G,f)) := \Im \{ \eta\from  S_{\e-\tau}(G,f) \to S_\e(G,f)\}.
    $$
\end{corollary}

\begin{proof}
This comes from combining \cref{lem:intersections_and_q} with the commutative diagram
\begin{equation*}
\begin{tikzcd}
    G \times [-\e+\tau, \e - \tau] \ar[r,hook] \ar[d,"q"]
        & G \times [-\e,\e] \ar[d, "q"]
    \\
    S_{\e-\tau}(G) \ar[r,"\eta"] & S_\e(G).
\end{tikzcd}
\end{equation*}
\end{proof}

\section{Properties of truncation}
\label{sec:properties_of_truncation}

In this section, we provide proofs of the main results stated in \cref{ssec:PropertiesOfTruncationSurvey}. 
We show the main results pertaining to connectedness of $S_\e^\tau(G,f)$ in \cref{ssec:apdx:Connectedness}, and give results for when $S_\e^\tau(G,f)$ is empty in \cref{ssec:apdx:Empty}

    \subsection{Connectedness}
    \label{ssec:apdx:Connectedness}
    Truncation does not necessarily preserve connectedness.
    In fact, for the ($0$-tailed) graph of \cref{fig:EmptyTauTruncation}, for any $\tau>0$, each edge of $T^\tau(G,f)$ is a separate component.
    In this section, we utilize the notions of $s$-safe and $t$-tailed (\cref{defn:tauTailed_and_tauSafe}) to show that for a $\tau$-tailed graph, $T^\tau(G,f)$ remains connected (\cref{lem:staysConnected})
    Finally, we use this to show in \cref{cor:ConnectedSmoothingAndTruncating} that for certain ranges of $\tau$ relative to $\e$, $S_\e^\tau(G,f)$ maintains its connected components. 
    
    \longTailsCombined*
    
    \begin{proof}
        For the first statement, assume $(G,f)$ is $t$-tailed.  
        A point $y\in S^\e(G,f)$ is an up-fork only if it has multiple interior-disjoint positive-height up-paths.
        For such paths to be interior-disjoint, the preimage of $y$ must contain a point $(x,\e)\in(G,f)\times[-\e,\e]$ where $x$ is an up-fork in $G$.
        Therefore, $x$ has a down-path $\gamma$ of height $t$ in $G$.
        $S^\e$ transforms $\gamma$ into a height $t+2\e$ down-path of $y$ in $S_\e(G,f)$.
        Therefore, every up-fork has a $t+2\e$ long down-path.
        Symmetrically, every down-fork has a $t+2\e$ long up-path, so $S_\e(G,f)$ is $(t+2\e)$-tailed.
        
        For the second statement, assume $(G,f)$ is $s$-safe, and thus by definition, $s$-tailed. 
        By the first statement, $S_\e(G,f)$ is $(s+2\e)$-tailed and therefore $(s+\e)$-tailed.
        Moreover, some $x\in(G,f)$ has an up-path and down-path of height $s$.
        $S_\e$ transforms their union to a monotone path of height $2s+2\e$, so $S_\e(G,f)$ contains a point with both an up-path and down-path of height $s+\e$.
        
        For the final statement, note that the empty Reeb graph is not $s$-safe for any $s$, whereas every nonempty Reeb graph is at least $0$-safe.
        Moreover, every Reeb graph, including the empty Reeb graph is at least $0$-tailed.
        So, setting $t=s = 0$ in the first two statements  gives the following corollary.
    \end{proof}

    The next lemma shows that even without smoothing first, truncation can preserve $\bullet$-safe and $\bullet$-tailed properties for decreased parameters.
    \singleEdge*
    \begin{proof}
    Let $x$ be a downfork in $T^\tau(G,f)\subseteq (G,f)$. 
    Then it is also a downfork in $(G,f)$, so if $(G,f)$ is $\e$-tailed, then $x$ has an up-path of height $\e$ in $(G,f)$. 
    Assuming the path is parameterized with respect to function value, the portion of this up-path defined on $[0,\e-\tau]$ consists of points who all have an up-path of height $\tau$ in $(G,f)$, so $x$ has an up path of height $\e-\tau$ in $T^\tau(G,f)$. 
    Showing up-forks have a down-path of height $\tau$ is a symmetric argument, so $T^\tau(G,f)$ is $(\e-\tau)$-safe. 
    
    If $(G,f)$ is $\e$-safe, then it just remains to show that $T^\tau(G,f)$ has a point with an up- and down-path of height $\e-\tau$. 
    Since $(G,f)$ has a point $x$ with an up- and down-path of height $\e$, a similar argument to the above shows that the portion of these parameterized paths  inside $[0,\e-\tau]$ remain in $T^\tau(G,f)$, so $T^\tau(G,f)$ is $(\e-\tau)$-safe.
    \end{proof}

    We will now characterize the structure of the points removed by truncating, as well as the structure of the remainder.
    In the following lemmas and utilizing the notation of \cref{defn:truncation}, let $U=U_\tau(G,f)$ be the set of points in $G$ with no up-path of height $\tau$
    and $D=D_\tau(G,f)$ be those without a height $\tau$ down-path.
    \begin{lemma}\label{lem:forestUD}
        If $(G,f)$ is $\tau$-tailed, then any point in $U$ roots an up-tree of height less than $\tau$, and any point in $D$ roots a down-tree of height less than $\tau$.
    \end{lemma}
    \begin{proof}
        Pick $x\in U$, and let $C$ be the set of points reachable from $x$ by (possibly height $0$) up-paths.
        As $x$ has no up-path of height $\tau$, $C$ has height less than $\tau$.
        It remains to show that $C$ is an up-tree.
        Suppose not, then $C$ contains a down-fork $x'$ of $(G,f)$.
        Since $(G,f)$ is $\tau$-tailed, $x'$ has an up-path of height $\tau$ that also lies in $C$, contradicting that $C$ has height less than $\tau$.
        So any point in $U$ roots an up-tree of height less than $\tau$.
        Symmetrically, any point in $D$ roots a down-tree of height less than $\tau$.
    \end{proof}
    
    \begin{lemma}\label{lem:disconnectedUD}
        If $(G,f)$ is connected and $\tau$-safe, then $U\cup D$ contains no path from $U$ to $D$.
    \end{lemma}
    \begin{proof}
        Suppose that $U$ and $D$ intersect, then there exists some $x\in U\cap D$.
        As $(G,f)$ is $\tau$-safe, it contains a point $y\notin U\cup D$.
        Since $(G,f)$ is connected, there is a simple path from $x$ to $y$.
        Because $x$ roots both an up-tree and a down-tree, this path must be monotone.
        Therefore, $y$ lies in the up-tree or down-tree rooted at $x$, so $y\in U\cup D$, which is a contradiction.
        So $U$ and $D$ are disjoint, and since $U$ and $D$ are open, there is no path inside $U\cup D$ from $U$ to $D$.
    \end{proof}
    
    For a component of $C$ of $U$ or $D$, call a point $x\in T^\tau(G,f)$ a \emph{root} of $C$ if it lies in the closure of $C$.
    \begin{lemma}\label{lem:atMostOneRoot}
        If $(G,f)$ is connected and $\tau$-tailed, then any component $C$ of $U\cup D$ has at most one root.
    \end{lemma}
    \begin{proof}
        If $(G,f)$ is not $\tau$-safe, then $T^\tau(G,f)$ is empty, so $U\cup D$ has no root and we are done.
        So assume that $(G,f)$ is $\tau$-safe.
        By \cref{lem:disconnectedUD}, $C\subseteq U$ or $C\subseteq D$.
        Without loss of generality, assume that $C\subseteq U$, and that $C$ has multiple roots $x,x'\notin U\cup D$.
        Let $\pi$ be a simple path in $C$ connecting $x$ and $x'$.
        Because $C$ contains no down-forks and $C$ contains all points reachable from $C$ by up-paths, $\pi$ has height $0$.
        Because $G$ is a Hausdorff space, this means that $x=x'$, so $C$ has at most one root.
    \end{proof}
    
    \begin{lemma}\label{lem:atMostOneRootPath}
        Let $\gamma$ be a simple path in $(G,f)$ that starts and ends in $T^\tau(G,f)$.
        If $C$ is a component of $U\cup D$ with at most one root, then $\gamma$ does not intersect $C$.
    \end{lemma}
    \begin{proof}
        Suppose instead that $\gamma(t)$ lies in $C$, and $\gamma$ starts and ends in $T^\tau(G,f)$, which is disjoint from $C$, there is some $t'<t$ and $t''>t$ for which $\gamma(t')$ and $\gamma(t'')$ is the unique root of $C$, contradicting that $\gamma$ is simple.
    \end{proof}
    
    \begin{corollary}\label{lem:safeSimplePath}
        If $(G,f)$ is $\tau$-tailed, then any simple path $\gamma$ in $(G,f)$ that starts and ends in $T^\tau(G,f)$ lies completely inside $T^\tau(G,f)$.
    \end{corollary}
    
    \begin{lemma}\label{lem:someRoot}
        If $(G,f)$ is connected and $\tau$-safe, then any component $C$ of $U\cup D$ has a root.
    \end{lemma}
    \begin{proof}
        Assume that $U\cup D$ is nonempty and $\tau >0$ (otherwise we are done).
        Then $U$ and $D$ are both nonempty.
        By Lemma~\ref{lem:disconnectedUD}, $U$ is disconnected from $D$, so $U\cup D$ has at least two components.
        Since $(G,f)$ is connected, there is a path from $C$ to a different component of $U\cup D$, so $C$ has at least one root.
    \end{proof}

    We show that connectivity is preserved for $\tau$-tailed graphs.
    \begin{proposition}\label{lem:staysConnected}
        If $(G,f)$ is connected and $\tau$-tailed, then $T^\tau(G,f)$ is also connected.
    \end{proposition}
    \begin{proof}
        Suppose not, then there is a simple path in $(G,f)$ connecting two components of $T^\tau(G,f)$.
        This path must enter and exit a component $C$ of $U\cup D$.
        Since $C$ has a single root, the path must revisit it, contradicting that the path is simple.
    \end{proof}

    Note that the previous collection of lemmas was independent of the smoothing operation. 
    We can combine \cref{lem:staysConnected} with \cref{prop:longTailsCombined} to get the main result of this section.

    \ConnectedSmoothingAndTruncating*
    
    \begin{proof}
     $G$ is connected by assumption, and by \cref{prop:longTailsCombined}, $S_\e(G,f)$ is $2\e$-tailed.
    This further implies that $S_\e(G,f)$ is $\tau$ tailed for every $\tau \in [0,2\e]$, so by \cref{lem:staysConnected}, $S_\e^\tau(G,f)$ is connected.
    \end{proof}

\subsection{Emptiness}
\label{ssec:apdx:Empty}
    We conclude this section with a results on how smoothing and truncation can affect the image of the Reeb graph  $\Im(G,f) := f(G) \subseteq \R$.
    It is relatively immediate to see that if  the diameter of the image $\|\Im(G,f)\| < 2\tau$, then $T^\tau(G,f)$ is empty. 
    However, when $\|\Im(G,f)\| \geq 2\tau$, it is still possible for the $\tau$ truncation to be empty; see, for example, \cref{fig:EmptyTauTruncation}.

    \begin{proposition}
    \label{prop:ImageOfTruncatedGraph}
        Fix a Reeb graph $(G,f)$ with $f(G) = [a,b] \subseteq \R$. 
        If $b-a < 2\tau$, then $\Im(T^\tau(G,f)) = \emptyset$. 
        Otherwise, $\Im(T^\tau(G,f)) \subseteq [a+\tau, b-\tau]$.
    \end{proposition}
    \begin{proof}
        First, if $b-a < 2\tau$, then we will show that $T^\tau(G,f)$ is empty. 
        Indeed, since any point in $T^\tau(G,f)$ is a point in $x$ with up-and down-paths of height $\tau$ in $G$, the projection of these paths imply that $[f(x)-\tau, f(x)+\tau] \subseteq \Im(G,f)$, but this is impossible if $b-a <2\tau$. 
        
        Now we can assume that $b-a \geq 2\tau$ and need to show that $\Im(T^\tau(G,f)) \subseteq [a+\tau, b-\tau]$. 
        Let $c \in \Im(T^\tau(G,f)) $, so there is an $x \in G$ with an up-path $\pi_+$ and a down path $\pi_-$, each of height $\tau$, with $\pi_+(0) = \pi_-(0) = x$.
        Then $f(\pi_+(1)) = c+\tau$, and since $f(\pi_+(1)) \leq b$, $c \leq b-\tau$. 
        The symmetric argument gives us that $f(\pi_-(1)) = c-\tau$ so $c \geq a+\tau$.
    \end{proof}

    We can mitigate the undesirable properties of truncation by smoothing first.
    In the case that we start with a connected Reeb graph, we can make a stronger statement about the change in the image when smoothing and truncating.

    \ImageOfSmoothingAndTruncating*

    \begin{proof}
    Note that by \cref{prop:ImageOfSmoothedGraph}, $\Im(S_\e(G,f)) = [a-\e, b+\e]$. 
    If we assume $b-a < 2(\tau-\e)$, then $(b+\e) - (a-\e) < 2\tau$, so by \cref{prop:ImageOfTruncatedGraph}, 
    $\Im(S_\e^\tau(G,f)) = \Im(T^\tau(S_\e(G,f))) = \emptyset$.
    
    Now, we can assume $b-a \geq 2(\tau-\e)$. 
    One direction of containment is easy since by \cref{prop:ImageOfSmoothedGraph}, 
    \begin{equation*}
        \Im(S_\e^\tau(G,f)) = \Im(T^\tau(S_\e(G,f))) \subseteq [a- (\e - \tau), b+ (\e-\tau)].
    \end{equation*}
    Thus, it remains to show that $[a- (\e - \tau), b+ (\e-\tau)] \subseteq \Im(S_\e^\tau(G,f))$.
    There exist points $s,t\in S_\e(G,f)$ with $f(s)=a-\e$ and $f(t)=b+\e$ that are connected by a simple (not necessarily monotone path) $\pi$.
    If $\pi$ is monotone, then truncating retains its monotone subpath with image $[a-(\e-\tau),b+(\e-\tau)]$, and we are done.
    If on the other hand $\pi$ is not monotone, consider its maximal monotone subpaths.
    The subpath containing $s$ ends at a downfork $d$, and the other subpath containing $d$ ends at an upfork $u$ with $f(d)\geq f(u)$.
    By \cref{prop:longTailsCombined}, every down-fork of $S_\e(G,f)$ has a $2\e$ long up-path, so $f(d)\leq b-\e$ and symmetrically $f(u)\geq a+\e$.
    Let $d'$ and $u'$ be the points reachable from $d$ and $u$ by a $2\e$ long up-path and down-path, respectively.
    Then, $f(d')\geq f(d)+\e\geq f(u)\geq f(u')-\e\geq a=f(s) + \e$, so the first monotone subpath of $\pi$ has height at least $4\e$, and truncating it retains a point $p$ with $f(p)=a-(\e-\tau)$.
    Symmetrically, the monotone subpath of $\pi$ containing $t$ retains a point $q$ with $f(q)=b+(\e-\tau)$ after truncation.
    By Proposition~\ref{cor:ConnectedSmoothingAndTruncating}, $p$ and $q$ are connected, so $[a- (\e - \tau), b+ (\e-\tau)] \subseteq \Im(S_\e^\tau(G,f))$ by the intermediate value theorem.
    \end{proof}

\section{Maps and their properties}
\label{sec:mapsAndProperties}

We next build several maps, closely related to the map $\eta:S_\e(G,f) \to S_{\e'}(G,f)$ coming from the smoothing, which we will use to eventually build a new flow for our category $\Reeb$.
We then show that truncation itself is functorial and relate its definition to these maps.

\subsection{The maps \texorpdfstring{$\eta$}{e}, \texorpdfstring{$\nu$}{v}, and \texorpdfstring{$\omega$}{w}}
\label{ssec:maps}

We will  define several maps that relate $S_\e^\tau(G,f)$ for various values of $\tau$ and $\e$.
These maps can be visualized in the $\e$-$\tau$ plane; see \cref{tab:Maps} for a handy reference of the maps and their requirements and properties.

The easiest map to define is $\nu$, which is simply an inclusion of an untruncated graph into the original.
That is, we have immediate inclusions maps $S^\tau_\e(G,f) \hookrightarrow S_\e(G,f)$ since $S^\tau_\e(G,f)$ is defined as a subspace of $S_\e(G,f)$.
Given $\tau \leq \tau'$, this generalizes to map $\nu\from S_\e^{\tau'}(G,f) \to S_\e^{\tau}(G,f)$ on the truncated smoothings.
The easiest way to see this is using the fact that any point with an up- and down-path of height $\tau'$ certainly has each of height $\tau$.

\begin{table}[h]
    \centering
\begin{tabular}{c|c cccc}
      &  Direction & Slope ($\frac{d\tau}{d\e}$) &
    \shortstack{Requirements for existence\\ $S_{\e_1}^{\tau_1}(G,f) \to S_{\e_2}^{\tau_2}(G,f)$ }
    & Notes & Defn.
    \\ \hline
    $\nu$& $\downarrow$ & vertical & $\e_1=\e_2$, $\tau_1\geq\tau_2$ & inclusion map& \ref{defn:nu}\\
    $\eta$ & $\to$ & horizontal & $\e_1\leq\e_2$, $\tau_1=\tau_2$ & smoothing map & \ref{defn:eta}\\
    $\omega$ & $\searrow$& $m \in [-\infty, 0]$  & $\e_1\leq\e_2$, $\tau_1\geq\tau_2$ & Lem~\ref{lem:squareCommutes}: 
    $\omega = \nu \eta = \eta \nu$ 
    & \ref{defn:omega}\\
    $\rho$& $\nearrow$ &  $m \in [0,1]$ & $0 \leq \tau_2-\tau_1 \leq \e_2-\e_1 $ & equals $\eta$ when ${\tau_1} = {\tau_2}$ &\ref{defn:rho}
\end{tabular}
    \caption{Maps and when they exist. }
    \label{tab:Maps}
\end{table}

\begin{definition}
\label{defn:nu}
For any $0 \leq \tau_i \leq 2\e$ and $\tau_2 \leq \tau_1$, the map
$\nu\from S_{\e}^{\tau_1} (G,f) \hookrightarrow S_{\e}^{\tau_2} (G,f)  $
is given by inclusion.
\end{definition}

The next map is the restriction of  $\eta\from S_\e(G,f) \to S_{\e'}(G,f)$ to the truncated graphs with $\e \leq \e'$.
It is immediate that the following diagram
\begin{equation*}
    \begin{tikzcd}
    {\color{violet} G \times [-\e + \tau, \e]}
        \ar[dd,hook,violet] \ar[rr,hook, crossing over]
        \ar[to=L,dashed,start anchor={[xshift=5ex]}, "q"]
    &&
    {\color{violet} G \times [-\e' + \tau, \e']}
        \ar[to=R,dashed,start anchor={[xshift=5ex]}, "q"]
    \\
    &
    |[alias=L]|  S_\e(G)
        \ar[rr,crossing over, "\eta", near start]
        \ar[from = dddl, dashed,start anchor={[xshift=6ex]}, "q", near start]
    &&
    |[alias=R]|  S_{\e'}(G)
    \\
    {\color{black} G \times [-\e , \e]}
        \ar[rr,hook, crossing over] \ar[to=L, "q"]&&
    {\color{black} G \times [-\e', \e']}
        \ar[to=R, "q"]
        \ar[from = uu, violet, crossing over, hook]
    \\
    \\
    {\color{dkorange} G \times [-\e , \e-\tau]}
        \ar[uu,hook',dkorange]
        \ar[rr,hook, crossing over]
        &&
    {\color{dkorange} G \times [-\e', \e'-\tau]}
        \ar[uu,hook',dkorange]
        \ar[to=R,dashed,start anchor={[xshift=6ex]}, "q"]
    \end{tikzcd}
\end{equation*}
commutes since the black square was shown to be commutative in \cref{eq:Square}, and all other maps are either inclusions or defined by composition.
From this diagram, we can extend the definition of $\eta$ to the truncated graphs.

\begin{definition}
\label{defn:eta}
For $0 \leq \e \leq \e'$ and $0 \leq \tau \leq 2\e$, the map
$\eta\from  S_\e^\tau(G,f) \to S_{\e'}^\tau(G,f)$ is given by the restriction of $\eta\from S_\e(G,f) \to S_{\e'}(G,f)$ to $\nu(S_{\e}^\tau(G,f)) \subseteq S_\e(G,f)$.
\end{definition}

The maps $\eta$ and $\nu$ commute as shown in the following lemma.

\begin{lemma}
\label{lem:squareCommutes}
For any $0\leq\tau\leq\tau'\leq 2\e\leq 2\e'$, the following diagram commutes:
\begin{equation*}
    \begin{tikzcd}
        S_\e^{\tau'}(G,f) \ar[d,hook, "\nu"]\ar[r, "\eta"] & S_{\e'}^{\tau'}(G,f)\ar[d,hook, "\nu"]\\
        S_\e^\tau(G,f)      \ar[r,"\eta"] & S_{\e'}^\tau(G,f)
    \end{tikzcd}
\end{equation*}
\end{lemma}

\begin{proof}
Because $\eta$ is defined by restricting to the relevant subspaces, we can use the diagram 
\begin{equation*}
    \begin{tikzcd}
        S_\e^{\tau'}(G,f) \ar[d,hook, "\nu"]\ar[r, "\eta"] 
        \ar[dd,hook, "\nu"', bend right = 75]
        & S_{\e'}^{\tau'}(G,f)\ar[d,hook, "\nu"]
        \ar[dd,hook, "\nu", bend left = 75]
        \\
        S_\e^\tau(G,f)      \ar[r,"\eta"] 
        \ar[d, hook, "\nu"]
        & S_{\e'}^\tau(G,f)
        \ar[d,hook, "\nu"]
        \\
        S_\e(G,f)      \ar[r,"\eta"] & S_{\e'}(G,f)
    \end{tikzcd}
\end{equation*}
The bottom square and outside face commute by \cref{defn:eta}, and the outer triangles commute as they are simply inclusions. Thus the top square commutes as desired. 
\end{proof}

\begin{figure}
\centering
    \includegraphics[page=1, width = 1.7in]{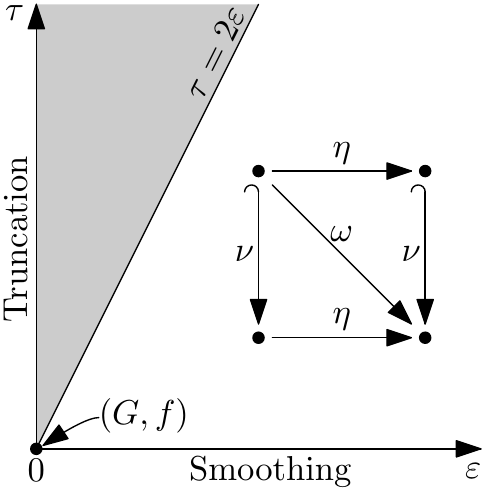}
    \caption{Visualization of the commutative diagram of \cref{lem:squareCommutes}.
    }
\label{fig:TruncatedMaps}
\end{figure}
We will use the square of \cref{lem:squareCommutes} repeatedly, so we name the diagonal map as follows.

\begin{definition}
\label{defn:omega}
For $0 \leq \tau_i \leq 2\e_i$ $i=1,2$, $\e_1 \leq \e_2$ and $\tau_2 \leq \tau_1$,
$\omega\from S_{\e_1}^{\tau_1}(G,f) \to S_{\e_2}^{\tau_2}(G,f)$ is defined by
is the diagonal of the square
\begin{equation*}
    \begin{tikzcd}
        S_{\e_1}^{\tau_1}(G,f)
            \ar[d,hook, "\nu"]
            \ar[r, "\eta"]
            \ar[dr,"\omega"]
        & S_{\e_2}^{\tau_1}(G,f)
            \ar[d,hook, "\nu"]\\
        S_{\e_1}^{\tau_2}(G,f)
            \ar[r,"\eta"]
        & S_{\e_2}^{\tau_2}(G,f).
    \end{tikzcd}
\end{equation*}
That is, $\omega = \eta \nu = \nu \eta$ when those maps exist.

\end{definition}

See \cref{fig:TruncatedMaps} for a visualization.
Note that this definition implies that $\omega = \eta$ if $\tau_1 = \tau_2$; and $\omega = \nu$ if $\e_1 = \e_2$.
Notice that we can also define $\omega$ by looking backwards by using \cref{prop:DefnBackwardView}.

\begin{lemma}
\label{lem:omega_restricts}
Assume $\tau_1$, $\tau_2$, $\e_1$ and $\e_2$ chosen so that the map $\omega\from S_{\e_1}^{\tau_1}(G,f) \to S_{\e_2}^{\tau_2}(G,f)$
is defined; i.e.~$0 \leq \tau_i \leq 2\e_i$,  $\e_1 \leq \e_2$ and $\tau_2 \leq \tau_1$.
\begin{enumerate}%
    \item The map $\omega$ is equal to the restriction of
    $\eta\from S_{\e_1}(G,f) \to S_{\e_2}(G,f)$ to
    $\nu(S_{\e_1}^{\tau_1}(G,f)) \subseteq S_{\e_1}(G,f)$.
    \item If $0 \leq \tau_i\leq \e_i$, then $\omega$ is the restriction of
        $\eta\from S_{\e_1}(G,f) \to S_{\e_2}(G,f)$ to
        $\eta(S_{\e_1-\tau_1}) \subseteq S_{\e_1}(G,f)$.
\end{enumerate}
\end{lemma}

\begin{proof}
The first statement is saying that the diagram
\begin{equation*}
    \begin{tikzcd}
        S_{\e_1}^{\tau_1}(G,f)
            \ar[d,hook, "\nu"]
            \ar[dr,"\omega"]
        &
            \\
        S_{\e_1}^{\tau_2}(G,f)
            \ar[r,"\eta"]
            \ar[d,hook, "\nu"]
        & S_{\e_2}^{\tau_2}(G,f)
            \ar[d,hook, "\nu"]
        \\
        S_{\e_1}(G,f)
            \ar[r,"\eta"]
        & S_{\e_2}(G,f)
    \end{tikzcd}
\end{equation*}
commutes.
The second statement is immediate from combining \cref{prop:DefnBackwardView},
(which requires $0 \leq \tau_i \leq \e_i$) with the first statement.
\end{proof}

\subsection{A categorical view of truncation }
\label{ssec:categorical_truncation}

    We next investigate the properties of truncation that arises from our combinatorial interpretation, \cref{prop:defn_combinatorial}, separating the truncation operation from the smoothing functor and considering it as an operation on a graph equipped with a height function in its own right.

    From \cref{defn:truncation}, it is immediate that on the objects, $T^0(G,f)=(G,f)$.
    It is also easy to see that $T^{\tau_1}(G,f) \subseteq T^{\tau_2}(G,f)$ for $\tau_1 \geq \tau_2$.
    To make it clear which graph we are working with, we will again use $\nu$ to denote the relevant inclusion.

    \begin{definition}
        For $\tau_1\geq \tau_2$, we write $\nu\from T^{\tau_1}(G,f) \hookrightarrow T^{\tau_2}(G,f)$ for the inclusion.
    \end{definition}

    \begin{proposition}
        For $0 \leq \tau_1,\tau_2$,  $T^{\tau_2}T^{\tau_1}(G,f) =  T^{\tau_1+\tau_2}(G,f)$.
    \end{proposition}

    \begin{proof}
        Let $x \in T^{\tau_2}T^{\tau_1}(G,f)$.
        Then $x$ has an up-path of height $\tau_2$ in $T^{\tau_1}(G,f)$, $\pi_2\from [0,1] \to G$.
        Since the endpoint, $\pi(1)$ is in $T^{\tau_1}(G,f)$, it has an up-path $\pi_1$ of height $\tau_1$ in $G$.
        Concatenating these gives an up-path in $G$ of height $\tau_1 + \tau_2$, so $x \in T^{\tau_1+\tau_2}(G,f)$.

        For the other direction, set $\tau = \tau_1+\tau_2$ and assume $x \in T^{\tau}(G,f)$.
        Then $x$ has an up-path  of height $\tau$, reparameterized as   $\pi\from  [0,\tau] \to G$ with $f(\pi(t)) = f(x) + t$.
        So $x \in T^{\tau_1}(G,f)$ since it has a path of height $\tau_1 \leq \tau$ in $G$ by simply restricting $\pi$ to $[0,\tau_1]$.
        Further, the path $\pi$ restricted to $[0,\tau_2]$ is an up-path of height $\tau_2$ starting from $x$ which is entirely contained in $T^{\tau_1}(G,f)$.
        This is because every point $\pi(t)$ for $t \in [0,\tau_2]$ has the up-path of height $\tau_1$ given by $\pi\from [t,t+\tau_1]$.
        Taken together, this means that $x \in T^{\tau_2}T^{\tau_1}(G,f)$.
    \end{proof}

    We next show that truncation $T^\tau$ is a functor on Reeb.
    To do this, we need to provide a morphism $T^\tau[\phi]: T^\tau(G,f) \to T^\tau(H,h)$ for a given morphism $\phi:(G,f) \to (H,h)$. 
    Indeed, we can simply use the restriction map, which is well defined by the following lemma. 

\begin{lemma}
\label{prop:mapRestrictsTruncation}
    Given a function preserving map $\phi:(G,f) \to (H,h)$, assume $T^\tau(G,f)$ is not empty. 
    Then the map $\phi$ restricts to $T^\tau(G,f)$, i.e.
    \begin{equation*}
    \begin{tikzcd}
       T^\tau(G,f) 
        \ar[r,dashed, "{\tilde \phi}"] 
        \ar[d, hook, "\nu"]
       & T^\tau(H,h)
        \ar[d, hook, "\nu"]
        \\
       (G,f) \ar[r, "\phi"] & (H,h)
    \end{tikzcd}
    \end{equation*}
    commutes, where $\tilde \phi(a) = \phi(a)$. 
\end{lemma}

\begin{proof}
Assume that $x \in T^\tau(G,f)$, so there is an up and down path in $(G,f)$. 
We can parameterize this by $\gamma:[-\tau,\tau] \to G$ so that $\gamma(0) = x$, and $f(\gamma(t)) = f(x) + t$. 
Then $\tilde \gamma:= \phi \gamma$ is a path in $(H,h)$. 
Further, because $\phi$ is function preserving, 
$g(\tilde \gamma(t)) = f\gamma(t) = f(x) + t$.
So $\tilde \gamma$ is  an up-and down path in $(H,h)$ of height $\tau$ with $\tilde \gamma(0) = \phi(x)$. 
Thus $\phi(x) \in T^\tau(H,h)$. 
\end{proof}

Now, we can recast this structure in the following lemma. 

    \begin{proposition}
    \label{lem:endofunctorTrunc}
        For any fixed $\tau\geq 0$, $T^\tau \in \End(\Reeb)$.
    \end{proposition}
    \begin{proof}
    Because $T^\tau[\phi]$ is well defined by \cref{prop:mapRestrictsTruncation}, we need only check that to check that $T^\tau[\Id] = \Id$ and that $T^\tau[\phi \circ \psi] = T^\tau[\phi] \circ T^\tau[\psi] $.
    However, these are also immediate from the definition by restriction, so we are done.

    \end{proof}

    Now, we can expand our view of $S_\e^\tau$ to be a functor since it is the composition of two functors $T^\tau S_\e$.
    In particular, we have the following definition.

    \begin{definition}
    For a morphism $\phi\from (G,f) \to (H,h)$, define $S_\e^\tau[\phi]\from  S_\e^\tau(G,f) \to S_\e^\tau(H,h)$ to be the restriction of $S_\e[\phi]\from S_\e(G,f) \to S_\e(H,h)$
    (given by \cref{eq:S_e_appliedtomaps} restricted to $\nu(S_\e^\tau[\phi])$).
    \end{definition}
    This definition is well defined since, as with $T^\tau[\phi]$, the image of a monotone path under $\phi$ is a monotone path, and thus $S_\e^\tau[\phi](S_\e^\tau(G,f)) \subseteq S_\e^\tau(H,h)$.
    With this definition for morphisms, the following lemma is immediate.

    \begin{lemma}
    \label{lem:endofunctor}
    For any fixed $\e$ and $\tau$ for which $0 \leq \tau \leq  2\e$, $S_\e^{\tau} \in \End(\Reeb)$.
    \end{lemma}

    Our final lemma shows how $S_\e^\tau$ interacts with the inclusion maps $\nu$.

    \begin{lemma}
    \label{lem:Square_S_nu}
    For $0 \leq \e$ and $0 \leq \tau' \leq \tau \leq 2\e$, the square 
    \begin{equation*}
        \begin{tikzcd}
            (G,f) \ar[d,"\phi"]
            & S_\e^\tau(G,f)
                \ar[d, "{S_\e^{\tau}[\phi]}"']
                \ar[r, hook, "\nu"]
            & S_{\e}^{\tau'}(G,f)
                \ar[d, "{S_\e^{\tau'}[\phi]}"]
            \\
           (H,h)
            & S_\e^{\tau}(H,h)
                \ar[r, hook, "\nu"]
            & S_{\e}^{\tau'}(H,h)
        \end{tikzcd}
    \end{equation*}
    commutes.
    \end{lemma}

    \begin{proof}
    The lemma follows from the fact that both $S_\e^\tau[\phi]$ and $S_{\e}^{\tau'}[\phi]$ are given by restricting $S_\e[\phi]$ to the relevant subset of $S_\e(G,f)$.
    That is, the remaining faces of the triangular prism
    \begin{equation*}
        \begin{tikzcd}
             S_\e^\tau(G,f)
                \ar[dd, "{S_\e^{\tau}[\phi]}"']
                \ar[dr, hook, "\nu"]
                \ar[rr, hook, "\nu"]
            &&
            S_\e(G,f)
                \ar[dd, "{S_\e[\phi]}"]
            \\
            & S_{\e}^{\tau'}(G,f)
                \ar[ur, hook, "\nu"]
            \\
            S_\e^{\tau}(H,h)
                \ar[dr, hook, "\nu"]
                \ar[rr, hook, "\nu", near start]
            &&
            S_\e(H,h)
            \\
            & S_{\e}^{\tau'}(H,h)
                \ar[ur, hook, "\nu"]
                \ar[from = uu, "{S_\e^{\tau'}[\phi]}", crossing over, very near start]
        \end{tikzcd}
    \end{equation*}
    commute, hence the left square commutes.
    \end{proof}

\section{Technical proofs for morphisms}
\label{sec:IsomorphismIssues}

In this section we fully unpack the various isomorphisms which were glossed over in the main body.  We begin with the related work, and in particular \cref{remark:abusingNotation}, which involves prior work.  We then follow these issues into our own proof that smoothing and truncation commute for $\tau$-safe graphs, which in turn implies that truncated smoothing can be applied to (as a functor) the morphims we presented in \cref{tab:Maps}.

\subsection{Isomorphism for smoothing}
\label{ssec:apdx:isomorphism}
The first technicality we note is the isomorphism issue noted in \cref{remark:abusingNotation}. 
The problem is that $S_\e S_\delta(G,f)$ is not exactly the same thing as $S_{\e+\delta}(G,f)$, although they are isomorphic. 
Basically, points in $S_{\e+\delta}(G,f)$ are equivalence classes of points from $G \times [-\e-\delta,\e+\delta]$, while points of $S_\e S_\delta(G,f)$ are equivalence classes of points in $S_\delta(G,f) \times [-\e,\e]$.

As noted in \cite[Obs.~4.30]{deSilva2016}, there is a natural isomorphism $S_{\e_2} S_{\e_1}(G,f) \to S_{\e_1+\e_2}(G,f)$ given as follows. 
Setting $\e = \e_1+\e_2$, a point in $S_{\e_2} S_{\e_1}(G,f)$ is represented by $q_2(q_1(x,\e_1),\e_2)$ and is sent  to $q_3(x,\e) \in S_{\e}(G,f)$ 
  making the diagram 
    \begin{equation*}
     \begin{tikzcd}[column sep= tiny]
    & (G \times [-\e_1,\e_1], f + \Id) 
        \ar[dr, "q_1"]
    &&
    (S_{\e_1}(G,f) \times [-\e_2,\e_2], f_{\e_1} + \Id) 
        \ar[dr, "q_2"]
        \\
    (G,f) 
        \ar[ur,hook,"{(\Id,0)}"] 
        \ar[drr,hook,"{(\Id,0)}"'] 
        \ar[drrrr, "\eta" description]
        \ar[rr, "\eta"]
    & &
    S_{\e_1}(G,f)
        \ar[ur,hook,"{(\Id,0)}"] 
        \ar[rr, "{S_{\e_2}[\eta]}"] 
     & &
    S_{\e_2} S_{\e_1}(G,f)
    \ar[d, "\cong"]
    \\
    &&
    (G \times [-\e,\e], f_{\e}) 
        \ar[rr, "q_3"']
    &&
    S_{\e_1+\e_2}(G,f)
    \end{tikzcd}
  \end{equation*}
  commute.

If we use the notation $S_\e(G,f)=: (G_\e,f_\e)$, constructing \cref{eq:S_e_appliedtomaps} with $\phi = \eta: (G,f) \to S_{\delta}(G,f)$
 gives the following commutative diagram 
     \begin{equation*}
  \begin{tikzcd}
    & (G \times [-\e,\e], f+\Id) 
        \ar[dr, "q"] 
        \ar[dd, violet,"{(\eta,\mathrm{Id})}", near start]
    \\
    (G,f) 
        \ar[ur,hook,"{(\mathrm{Id},0)}"] 
        \ar[rr, "\eta", near start, crossing over] 
        \ar[dd,violet, "\eta" ]
    & &
    S_\e(G,f) \ar[dd, violet, "{S_\e[\eta]}"]
    \\
    & 
    (G_\delta \times [-\e,\e], f_\delta + \Id) 
        \ar[dr, "q"] 
    \\
    S_\delta(G,f)
        \ar[ur,hook,"{(\mathrm{Id}_G,0)}"] 
        \ar[rr, "\eta"] 
        & &
    S_{\e}(S_\delta(G,f)).
  \end{tikzcd}
  \end{equation*}
  So, we abuse notation and write $\eta: S_{\e_1}(G,f) \to S_{\e_1+\e_2}(G,f)$ for the composition of maps 
  \begin{equation*}
      \begin{tikzcd}
        S_{\e_1}(G,f) \ar[d, "{S_{\e_2}[\eta]}"']
          \ar[dr, "\eta"]\\
        S_{\e_2}(S_{\e_1}(G,f)) \ar[r,"\cong"] & S_{\e_1+\e_2}(G,f)
      \end{tikzcd}
  \end{equation*}
  giving a map $\eta:S_\e(G,f) \to S_{\e'}(G,f)$ for every $\e \leq \e'$.

\subsection{Safe truncation and smoothing commute}
\label{ssec:apdx:commutes}

    We next construct a map $\psi\from S_\e T^\tau(G,f)\to T^\tau S_\e(G,f)$ in \cref{lem:conflow0} which is continuous and function-preserving; we then show that if $(G,f)$ is $\tau$-safe, $\psi$ is an isomorphism, proving that truncation and smoothing commute (\cref{prop:conflow01}) and smoothing and truncate combine additively (\cref{thm:additive}).

    Let~$\sim$ and $\sim_T$ respectively be the relevant equivalence relations on $(G,f)\times[-\e,\e]$ and $T^\tau(G,f)\times[-\e,\e]$, and let $q$ and $q_T$ be their quotient maps.
    
    \begin{lemma}\label{lem:conflow0}
        Let $\nu$ and $\nu_S$ be the inclusion maps of the relevant truncations.
        For any $0\leq\tau$ and $0\leq\e$, there exists a unique map $\psi\from S_\e T^\tau(G,f)\to T^\tau S_\e(G,f)$ that makes the diagram
        \begin{equation*}
        \begin{tikzcd}
            T^\tau(G,f)\times[-\e,\e]\ar[d,hook',swap,"{(\nu,\id)}"]\ar[r,"q_T",two heads] & S_\e T^\tau(G,f)\ar[r,"\psi"] & T^\tau S_\e(G,f)\ar[d,hook',swap,"\nu_S"]\\
            (G,f)\times[-\e,\e]\ar[rr,"q",two heads] && S_\e(G,f)
        \end{tikzcd}
        \end{equation*}
        commute, and $\psi$ is continuous and function-preserving.
    \end{lemma}
    \begin{proof}
        Since $q_T$ is surjective and $\nu_S$ is injective, the map $\psi$, if it exists, is uniquely determined by the above diagram.
        Any function $\psi$ making this diagram commute is automatically continuous and function-preserving.
        
        Let $\psi(q_T(x,\lambda)) := q(x,\lambda)$.
        We first show that the image of $\psi$ lies in $T^\tau S_\e(G,f)$.
        By surjectivity of $q_T$, it suffices to show for any $x\in T^\tau(G,f)$ and $\lambda\in[-\e,\e]$, that $q(x,\lambda) \in T^\tau S_\e(G,f)$.
        We have some height $\tau$ up-path $\pi$ in $G$ starting at $x$.
        Then $s\mapsto q(\pi(s),\lambda)$ is a height $\tau$ up-path of $q(x,\lambda)$ in $S_\e(G,f)$, and $q(x,\lambda)$ similarly has a height $\tau$ down-path, so $q(x,\lambda)\in T^\tau S_\e(G,f)$.
        It remains to show that $\psi$ is well-defined by showing that whenever $q_T(x,\lambda)=q_T(x',\lambda')$, we have $q(x,\lambda)=q(x',\lambda')$.
        This follows from the inclusion of any equivalence class of $\sim_T$ into one of $\sim$.
        Explicitly, any path with constant function value in $T^\tau(G) \times [-\e,\e]$ is contained in $G \times [-\e,\e]$, where it also has constant function value.
    \end{proof}

    We will show that $S_\e^\tau[\eta]$ is interchangeable with $\eta$ when after accounting for the isomorphism.
    
        \begin{theorem}
        \label{thm:Ttaunu}
            If $0\leq\tau$ and $0\leq\e$, and $\eta$ and $\eta_T$ are the natural maps into the smoothings in the following diagram, the diagram commutes.
            \begin{equation*}
            \begin{tikzcd}
                T^\tau(G,f) \ar[r,"\eta_T",swap]\ar[rr,"{T^\tau[\eta]}",bend left=10]
                & S_\e T^\tau(G,f) \ar[r,"{\psi}",swap] & T^\tau S_\e(G,f)
                \\
                (G,f) \ar[rr,"{\eta}"] & & S_\e(G,f).
            \end{tikzcd}
            \end{equation*}
        \end{theorem}
        \begin{proof}
            Recall that $T[\eta]$ is simply the restriction of $\eta$ to the truncation.
            Commutativity of the diagram below is evident from Lemma~\ref{lem:conflow0}, and the theorem follows.
            \begin{equation*}
            \begin{tikzcd}
                T^\tau(G,f) \ar[rr,"{\eta_T\from x\mapsto q_T(x,0)}",swap]\ar[rrr,"{T^\tau[\eta]=T^\tau[x\mapsto q(x,0)]}",bend left=10]\ar[dr,"{(\id,0)}",hook',swap]
                \ar[dd,"\nu",hook',swap] 
                & &
                S_\e T^\tau(G,f) \ar[r,"{\psi=q\circ(\nu,\id)\circ q_T\inv}",swap]
                &[4em]
                T^\tau S_\e(G,f)
                \ar[dd,"\nu_S",hook',swap]
                \\
                &
                T^\tau(G,f)\times[-\e,\e] 
                \ar[ur,"q_T",swap,two heads]
                \\
                (G,f) \ar[rrr,"{\eta\from x\mapsto q(x,0)}",crossing over,swap]\ar[dr,"{(\id,0)}",hook',swap] & & &
                S_\e(G,f).\\
                &
                (G,f)\times[-\e,\e] \ar[urr,"q",swap,two heads]
                \ar[from = uu,"{(\nu,\id)}",hook',swap,near start, crossing over]
            \end{tikzcd}
            \end{equation*}
        \end{proof}

\conflow*

   We will break up the proof into several lemmas. 
    Assume that $(G,f)$ is $\tau$-safe.
    It is a standard result of point set topology that a continuous bijection from a compact space to a Hausdorff space is a homeomorphism \cite[Thm 26.6]{Munkres2000}.
    Because $\psi$ is continuous and function-preserving, it remains to show that it is a bijection.
    For injectivity, we use \cref{lem:safeSimplePath} of $\tau$-safe graphs that simple paths that start and end in $S_\e T^\tau(G,f)$ lie completely inside $S_\e T^\tau(G,f)$.
    \begin{lemma}
    \label{lem:psiInjective}
        If $(G,f)$ is $\tau$-safe, then $\psi$ is injective.
    \end{lemma}
    \begin{proof}
        Suppose it is not, then there exist points $y\neq y'\in S_\e T^\tau(G,f)$ with $\psi(y) = \psi(y')$.
        Without loss of generality, let $(x,\lambda),(x',\lambda')\in T^\tau(G,f)\times[-\e,\e]$ with $y=:q_T(x,\lambda)$ and $y'=:q_T(x',\lambda')$.
        As $y\neq y'$ and $\psi(y)=\psi(y')$, we have
        \begin{align*}
            (x,\lambda)&\not\sim_T (x',\lambda'), and\\
            (x,\lambda)&    ~\sim~ (x',\lambda').
        \end{align*}

        There is a simple function preserving path $\pi\from (x,\lambda) \rightsquigarrow (x',\lambda')$ in $(G,f) \times [-\e,\e]$.
        Let $(\gamma,h):=\pi$.
        We show that its first component $\gamma$ is also simple.
        Otherwise $\gamma(t) = \gamma(t')$ for some $t\neq t'$, and since $\pi$ is function preserving, $f(\pi(t)) = f(\pi(t'))$, then 
        \begin{equation*}
            f(\gamma(t)) + h(t) = f(\gamma(t')) + h(t') = f(\gamma(t)) + h(t'),
        \end{equation*}
        so $h(t) = h(t')$ and hence $\pi(t) = \pi(t')$ contradicting that $\pi$ is simple.
        So $\gamma$ is a simple path.

        Moreover, $\gamma$ starts and ends in $T^\tau(G,f)$, so because $(G,f)$ is $\tau$-safe, $\gamma$ lies completely inside $T^\tau(G,f)$ by \cref{lem:safeSimplePath}.
        So $\pi$ is a function preserving path in $T^\tau(G,f) \times [-\e,\e]$, and therefore $(x,\lambda)\sim_T(x',\lambda')$, which is a contradiction, so $\psi$ is injective.
    \end{proof}

    For surjectivity, we use the fact that truncated points lie in an up-tree or down-tree of height at most $\tau$.
    \begin{lemma}
    \label{lem:psiSurjective}
        If $(G,f)$ is $\tau$-safe, then $\psi$ is surjective.
    \end{lemma}
    \begin{proof}
        Recall that $\psi(q_T(x,\lambda)):=q(x,\lambda)$.
        Suppose for a contradiction that $\psi$ is not surjective.
        Then a point $y$ does not lie in the image of $\psi$, so $q\inv(y)$ does not intersect $T(G,f)\times[-\e,\e]$.
        Denote by $X_y$ the set of points $x$ with $(x,\lambda)\in q\inv(y)$ (for some $\lambda\in[-\e,\e]$).
        Then all points in $X_y$ are truncated.
        Because $(G,f)$ is $\tau$-safe, each connected component of $G\setminus T$ is a union of up-trees or a union of down-trees of height less than $\tau$.
        Since $X_y$ is connected, it lies completely inside one such component, and assume without loss of generality that it is a union of uptrees $U$.
        
        As $(G,f)$ is $\tau$-safe, the closure of $U$ contains exactly one point of $T$, call it $r$.
        We have $f(r)<f(u)\leq f(r)+\tau$  for any $u\in U$.
        Consider a point $x\in X_y$.
        Since $x\in U$, there is a down-path $\gamma\from x\rightsquigarrow r$.
        We have $q(x,f(y)-f(x))=y$ and as $f(\gamma(t))$ is nonincreasing, $q(\gamma(t),f(y)-f(\gamma(t)))=y$ as long as $f(\gamma(t))\geq f(y)-\e$.
        As $r\notin X_y$, we therefore have $f(r)<f(y)-\e$.

        As $y\in T^\tau S_\e(G,f)$, $y$ has an up-path of height $\tau$ in $S_\e(G,f)$ to a point $y'$.
        Then $f(y')=f(y)+\tau>f(r)+\tau+\e \ge f(u)+\e$ for any $u\in U$.
        For $x'\in X_{y'}$, we have $f(x')\geq f(y')-\e$, so $X_{y'}$ does not intersect $U$.
        Therefore, any path from $X_y$ to $X_{y'}$ passes through $r$.
        In particular, the up-path from $y$ to $y'$ contains a point $y''$ with $f(y)\leq f(y'')$ and $r\in X_{y''}$, but then $f(y)\leq f(y'')\leq f(r)+\e$, contradicting that $f(r)<f(y)-\e$.
        So $\psi$ is surjective.
    \end{proof}
    
    \begin{proof}[Proof of \cref{prop:conflow01}]
        We have constructed a continuous function-preserving map $\phi\from S_\e T^\tau(G,f)\to T^\tau S_\e(G,f)$.
        Moreover, we have shown in \cref{lem:psiInjective,lem:psiSurjective} that it is a bijection if $(G,f)$ is $\tau$-safe.
        
        A continuous bijection from a compact space to a Hausdorff space is a homeomorphism \cite[Thm 26.6]{Munkres2000}.
        A constructible space in the sense of \cite{deSilva2016} is compact and Hausdorff, and both Reeb graphs and the smoothing of a Reeb graph are constructible.
        Similarly, truncation preserves constructibility as it removes an open subset.
        Hence, the continuous bijection $\psi$ is a homeomoprhism.
        As it is also function preserving, we have that $S_\e T^\tau(G,f)\cong T^\tau S_\e(G,f)$ as Reeb graphs.
    \end{proof}

\subsection{Applying \texorpdfstring{$S_\e^{\tau}$}{Set} to morphisms}
\label{ssec:truncation_applied_to_maps}

We  will next use this map $\psi$ to show that truncated smoothing can be applied to morphisms from \cref{tab:Maps}.  
For example, $S^\tau_\e[\eta]$ is equivalent to a different version of $\eta$ (up to suppressing our isomorphism between $S_\e^\tau(S_{\e'}^{\tau'}(G,f))$ and $S_{\e+\e'}^{\tau+\tau'}(G,f)$), and similar statements hold for $\omega$ and $\rho$.
That is, this section is meant to justify the reuse of the symbol $\eta$ to represent the various maps $S_\e^\tau[\eta]$ arising from different values of $\e$ and $\tau$.
For this, we show that the natural inclusion $\eta'\from S_{\e+\e_1}^{\tau + \tau_1}(G,f) \to S_{\e+\e_2}^{\tau + \tau_1}(G,f)$ respects the isomorphism.

    \begin{lemma}
    \label{lem:Functorial_eta}
    For $0 \leq \tau_1,\tau_2$ and $0\leq \e_1\leq\e_2$, as well as $0 \le \tau \le 2 \e$ and smoothing maps $\eta$ and $\eta'$, the diagram 
    \begin{equation*}
        \begin{tikzcd}
            S_{\e_1}^{\tau_1}(G,f) \ar[d,"\eta"] 
            & S_\e^\tau(S_{\e_1}^{\tau_1}(G,f)) 
                \ar[d, "{S_\e^{\tau}[\eta]}"'] 
                \ar[r, "\cong"]
            & S_{\e+\e_1}^{\tau + \tau_1}(G,f) 
                \ar[d, "\eta'"] 
            \\
            S_{\e_2}^{\tau_1}(G,f)
            & S_\e^\tau(S_{\e_2}^{\tau_1}(G,f))
                \ar[r, "\cong"]
            & S_{\e+\e_2}^{\tau + \tau_1}(G,f) 
        \end{tikzcd}
    \end{equation*}
    commutes.
    That is, abusing notation by suppressing $\psi$, we have 
    $S_\e^\tau[\eta] = \eta'$.
    \end{lemma}
    
    \begin{proof}
    We use \cref{thm:Ttaunu}: the $\eta$ in the diagram is the $T^{\tau_1}[\eta\colon S_{\e_1}(G,f)\to S_{\e_2}(G,f)]$ map, 
    and by replacing $(G,f)$ by $S_{\e_1}(G,f)$ and $\tau$ by $\tau_1$, the lemma follows.
    \end{proof}

    Since the $\omega$ and $\rho$ maps are simply restrictions of $\eta$, we can apply this results to those maps as well. 
    \begin{corollary}
    \label{lem:Functorial_omega}
    For 
    $\omega\from  S_{\e_1}^{\tau_1}(G,f) \to S_{\e_2}^{\tau_2}(G,f)$, 
    the diagram 
    \begin{equation*}
        \begin{tikzcd}
            S_{\e_1}^{\tau_1}(G,f) \ar[d,"\omega"] 
            & S_\e^\tau(S_{\e_1}^{\tau_1}(G,f)) 
                \ar[d, "{S_\e^{\tau}[\omega]}"'] 
                \ar[r, "\cong"]
            & S_{\e+\e_1}^{\tau + \tau_1}(G,f) 
                \ar[d, "\omega"] 
            \\
            S_{\e_2}^{\tau_2}(G,f)
            & S_\e^\tau(S_{\e_2}^{\tau_2}(G,f))
                \ar[r, "\cong"]
            & S_{\e+\e_2}^{\tau + \tau_2}(G,f) 
        \end{tikzcd}
    \end{equation*}
    commutes so long as all spaces and maps exist.
    That is,  abusing notation by suppressing the homeomorphism, we have 
    $S_\e^\tau[\omega] = \omega$.
    \end{corollary}
    
    \begin{proof}
    This is immediate from combining \cref{lem:Functorial_eta} with \cref{lem:omega_restricts}.
    \end{proof}

    \begin{corollary}
    \label{lem:Functorial_rho}
    For 
    $\rho\from  S_{\e_1}^{\tau_1}(G,f) \to S_{\e_2}^{\tau_2}(G,f)$, 
    the diagram 
    \begin{equation*}
        \begin{tikzcd}
            S_{\e_1}^{\tau_1}(G,f) \ar[d,"\rho"] 
            & S_\e^\tau(S_{\e_1}^{\tau_1}(G,f)) 
                \ar[d, "{S_\e^{\tau}[\rho]}"'] 
                \ar[r, "\cong"]
            & S_{\e+\e_1}^{\tau + \tau_1}(G,f) 
                \ar[d, "\rho"] 
            \\
            S_{\e_2}^{\tau_2}(G,f)
            & S_\e^\tau(S_{\e_2}^{\tau_2}(G,f))
                \ar[r, "\cong"]
            & S_{\e+\e_2}^{\tau + \tau_2}(G,f) 
        \end{tikzcd}
    \end{equation*}
    commutes so long as all spaces and maps exist.
    That is,  abusing notation by suppressing the homeomorphism, we have 
    $S_\e^\tau[\rho] = \rho$.
    \end{corollary}
    
    \begin{proof}
    This is immediate from combining \cref{lem:Functorial_eta} with \cref{defn:rho}.
    \end{proof}

\section{Categorical flow for \texorpdfstring{$\tau \in [0,\e]$}{t in [0,e]}}
\label{sec:categorical_flow}

Intuitively, we would like to be able to define a new categorical flow which incorporates the truncation parameter.
That is, a flow given increasing $\e$ parameter and non-decreasing $\tau$.
We have two barriers to deal with.

\begin{figure}
    \centering
    \includegraphics[page=2]{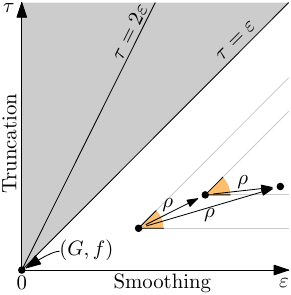}
    \caption{An illustration where $\rho$ is defined: given $S^{\tau_1}_{\e_1}$ and $S^{\tau_2}_{\e_2}$, the map $\rho$ exists if $(\e_2,\tau_2)$ is in the wedge that originates at point $(\e_1,\tau_1)$ with slope between 0 and 1 (shown in orange), and these maps commute whenever available.}
    \label{fig:ParametersRho}
\end{figure}
First, for some choices of parameters  truncating and then smoothing is not the same as smoothing and then truncating; i.e.~$S_\e^\tau S_\e^\tau \not \cong S_{2\e}^{2\tau}$.
For example, in \cref{fig:smoothingOperation_withTrunc}, $S_0^{2\e} S_{2\e}^{0}(G,f)$ has one connected component, but $ S_{2\e}^{0} S_0^{2\e} (G,f)  = S_{2\e}T^{2\e}(G,f)$ has two.
The second roadblock is that the forward morphisms $S_\e^\tau \Rightarrow S_{\e'}^{\tau'}$ required of a categorical flow corresponding to $\e \leq \e'$ might not be available since we require morphisms to be function preserving, and the truncation can erode the height of the codomain enough to leave no available points to map to.
Luckily, both of these issues can be overcome if we restrict the amount of truncation done to  $\tau \in [0,\e]$.
So, working in this restricted space, we will first investigate properties of $S_\e^\tau$, build the map $\rho$ and prove that $S_\e^\tau$ can be used to build a flow on $\Reeb$.

\subsection{The map \texorpdfstring{$\rho$}{p}}
\label{subsec:rho}

So far, all maps we have used go down and/or right in the $\e$-$\tau$ plane.
In general, we do not have a way to map up-right, that is for both increasing smoothing and truncating parameters.
This in fact makes sense: truncated graphs naturally include downward, as they are subsets of less-truncated versions, and it is not clear how to map from the larger, less-truncated/smoothed version to something that is more smoothed and more truncated.
However, we will see in this section that we do have such a map in a restricted setting.

Assume $0 \leq \e \leq \e'$, $0\leq \tau \leq 2\e$ and $0\leq \tau' \leq 2\e'$
so that
$S_\e^\tau(G)$ and $S_{\e'}^{\tau'}(G)$ exist.
We further assume $0\leq{\tau'-\tau}\leq{\e'-\e}$; that is, the slope of the line connecting $S_\e^\tau(G)$ and $S_{\e'}^{\tau'}(G)$ in parameter space is in $[0,1]$.
There are two equivalent ways to find the map
$\rho\from S_\e^\tau(G) \to S_{\e'}^{\tau'}(G)$.
First, with the assumptions we have that $0 \leq \e-\tau \leq \e'-\tau'$.
Using \cref{prop:DefnBackwardView}, we have that
    $$
    S_\e^\tau(G) = \eta(S_{\e-\tau}(G))
    : = \Im \{ \eta\from  S_{\e-\tau}(G) \to S_\e(G)\}.
    $$
so we have the commutative diagram
\begin{equation*}
    \begin{tikzcd}
        S_{\e-\tau}(G) \ar[rr, "\eta"] \ar[dr, "\eta"]
        &&  S_{\e'-\tau'}(G) \ar[dr, "\eta"]\\
        & S_\e(G) \ar[ur, leftrightarrow, "\eta"] \ar[rr, "\eta"] & & S_{\e'}(G)
    \end{tikzcd}
\end{equation*}
where the middle $\eta$ arrow points towards the larger of  $\e$ or $\e'-\tau'$.
Thus, we can define $\rho$ as follows, see also the visualization of  \cref{fig:ParametersRho}.
\begin{definition}
\label{defn:rho}
For $0 \leq \tau_i \leq \e_i$, $i=1,2$, and
$0 \leq \e_2 - \e_1 \leq \tau_2 - \tau_1$, we have a map
$\rho\from S_{\e_1}^{\tau_1}(G) \to S_{\e_2}^{\tau_2}(G)$
which is the restriction of
$\eta\from S_{\e_1}(G) \to S_{\e_2}(G)$ to $\eta(S_{\e_1-\tau_1}(G)) \subseteq S_{\e}(G)$.
\end{definition}

An equivalent construction of $\rho$ comes from utilizing \cref{cor:DefnIntersection} where $S_\e^\tau(G) = q(G \times [-\e + \tau, \e-\tau])$.
Because of the constraints,
$$
-\e'+\tau' \leq -\e + \tau \leq 0 \leq \e-\tau \leq \e'-\tau'
$$
and thus $G \times [-\e+\tau ,\e-\tau] \subseteq G \times [-\e'+\tau' ,\e'-\tau']$.
So, we can view $\rho$ as the restriction of $\eta$ to $q(G \times [-\e+\tau ,\e-\tau] )$ in the commutative diagram
\begin{equation*}
    \begin{tikzcd}
       G \times [-\e+\tau ,\e-\tau] \ar[r,hook] \ar[d,hook]
       & G \times [-\e'+\tau' ,\e'-\tau'] \ar[d,hook]\\
       G \times [-\e ,\e] \ar[r,hook] \ar[d,"q"]
       & G \times [-\e' ,\e'] \ar[d,"q"]\\
     S_\e(G) \ar[r, "\eta"] & S_{\e'}(G).
    \end{tikzcd}
\end{equation*}

The assumptions on $\e$ and $\tau$ mean that in terms of the visualization of \cref{fig:ParametersRho}, $\rho$ is available at slope between 0 and 1, denoted by the yellow wedge.
Further, these maps compose whenever available as shown below.

\begin{lemma}
\label{lem:triangles}
Whenever the required spaces and maps exist, the diagrams
\begin{equation*}
    \begin{tikzcd}%
        &S_{\e'}^{\tau'}(G)\ar[r, "\rho"]
        &S_{\e''}^{\tau''}(G)
        &S_{\e'}^{\tau'}(G)
        \ar[dr, "\omega"]
        & S_{\e}^{\tau}(G)
         \ar[dr, "\omega"]
        \ar[rr, "\rho"]
        &&S_{\e'}^{\tau'}(G)
        \\
        S_{\e}^{\tau}(G)
            \ar[ur, "\rho"]
            \ar[urr,  "\rho"]
        &&S_{\e}^{\tau}(G)
        \ar[rr, "\rho"]
        \ar[ur, "\rho"]
        &&S_{\e''}^{\tau''}(G)
        &S_{\e''}^{\tau''}(G)
        \ar[ur, "\rho"]
    \end{tikzcd}
\end{equation*}
commute.
\end{lemma}

\begin{proof}
The result arises from using the fact that both $\rho$ and $\omega$ are restrictions of the $\eta$ maps to the relevant subspaces of $S_{\e}(G,f)$.
\end{proof}

While we leave the full details to \cref{ssec:truncation_applied_to_maps}, we note that with the suppression of the homeomorphism discussed in \cref{remark:abusingNotation}, we have the following notation.

\textit{Notation:} Up to homeomorphism, $S_\e^\tau[\eta] = \eta$, $S_\e^\tau[\omega] = \omega$, and $S_\e^\tau[\rho] = \rho$.

That is, for an 
$\eta:S_{\e_1}^{\tau_1}(G,f) \to S_{\e_2}^{\tau_2}(G,f)$, 
$S_\e^\tau[\eta]: S_{\e_1+\e}^{\tau_1+\tau}(G,f) \to S_{\e_2+\e}^{\tau_2+\tau}(G,f)$
is the standard map $\eta$ from \cref{defn:eta}; the others are defined similarly from \cref{defn:omega,defn:rho}.

\subsection{The categorical flow}

We can use this information to construct a truncated flow based on any line with slope at most 1.
In particular, we can now prove our main result.

\CategoricalFlow*

Proving this theorem amounts to checking that $S^m$ satisfies quite a few properties encased in \cref{defn:Flow}, largely handled in the previous sections.
The last bits necessary are the following two lemmas.

\begin{lemma}
\label{lem:build_nat_trans}
For any   $0 \leq \e \leq \e'$, $0\leq \tau \leq 2\e$ and $0\leq \tau' \leq 2\e'$ satisfying $0\leq{\tau'-\tau}\leq{\e'-\e}$ there is a natural transformation given by
$\rho\from  S_\e^{\tau} \Rightarrow S_{\e'}^{\tau'}$;
in particular the diagram
\begin{equation*}
    \begin{tikzcd}
        G \ar[d,"\psi"]
        & S_\e^{\tau}(G)
            \ar[d, "{S_\e^{\tau}[\psi]}"']
            \ar[r, blue, "\rho"]
        & S_{\e'}^{\tau'}(G)
            \ar[d, "{S_{\e'}^{\tau'}[\psi]}"]
        \\
        H
        & S_\e^{\tau}(H) \ar[r,blue,"\rho" ]
        & S_{\e'}^{\tau'}(H)
    \end{tikzcd}
\end{equation*}
commutes.

\end{lemma}

\begin{proof}
    Combining the definition given by \cref{cor:DefnIntersection}
    \cref{eq:Se_of_morphism_and_eta_square}, we have that the diagram
\begin{equation*}
  \begin{tikzcd}
    G \times [-\e + \tau, \e-\tau]
        \ar[rrr, hook]
        \ar[dr,"q"]
        \ar[ddd,blue, "{(\psi,\mathrm{Id})}"]
    & & &
    G \times [-\e' + \tau', \e' - \tau']
        \ar[dr,"q"]
        \ar[ddd,blue, "{(\psi,\mathrm{Id})}"]
    \\
    &
    S_\e^\tau(G)
    \ar[rrr, "\rho",crossing over]
    \ar[dr,hook]
    & & &
    S_{\e'}^{\tau'}(G)
    \ar[dr,hook]
    \ar[ddd, blue, "{S_{\e'}^{\tau'}[\psi]}"]
    \\
     & &
     S_\e(G)
     \ar[rrr,"\eta",crossing over]
     & & &
     S_{\e'}(G)
     \ar[ddd, blue, "{S_{\e'}[\psi]}"]
     \\
    H \times [-\e + \tau, \e-\tau]
        \ar[rrr, hook]
        \ar[dr,"q"]
    & & &
    H \times [-\e' + \tau', \e' - \tau']
        \ar[dr,"q"]
    \\
    &
    S_\e^\tau(H) \ar[rrr, "\rho"]
    \ar[from = uuu,blue, "{S_\e^\tau[\psi]}", crossing over]
    \ar[dr,hook]
    & & &
    S_{\e'}^{\tau'}(H)
    \ar[dr,hook]
    \\
     & &
     S_\e(H) \ar[rrr,"\eta"]
     \ar[from = uuu, blue, "{S_\e[\psi]}", crossing over]
     & & &
     S_{\e'}(H)
  \end{tikzcd}
  \end{equation*}
  commutes, and thus the inner square commutes.
\end{proof}

\begin{proof}[Proof of \cref{thm:CategoricalFlow}]
\cref{lem:endofunctor} and \cref{lem:build_nat_trans} combine to show that $S^m$ is a functor.
By \cref{prop:conflow01}, $S_{\e'}^{m\e'}S_\e^{m\e} \cong S_{\e+\e'}^{m(\e+\e')}$ so $S^m$ is additive.
The combination means that $S^m$ is a categorical flow on $\Reeb$.
\end{proof}

\section{Full proof of strong equivalence of metrics}
\label{sec:appendix_strongequivmetrics}
Our final appendix is dedicated to the specifics of the proof of strong equivalence of metrics for some choices of $m$ and $m'$, with the result for all pairs $0 \leq m \leq m' <1 $ given as \cref{cor:equivalenceAll_m}.

\EquivalentMetrics*

\begin{proof}
\textbf{First inequality.} We first show that 
    \[
    d_I^{m}((G,f), (H,h)) 
    \leq  d_I^{m'}((G,f), (H,h)). 
    \]
Assume we have an $\e$-interleaving using $m'$ given by 
\begin{align*}
   \alpha\from  &(G,f) \to S_{\e}^{m'\e}(H,h)\\
   \beta\from  &(H,h) \to S_{\e}^{m'\e}(G,f).
\end{align*} 
Our goal is then to construct an $\e$-interleaving using $m$ by utilizing maps $\rho$ and $\omega$ that go between the lines $y=mx$ and $y=m'x$.
These can then be concatenated with the $\alpha$ and $\beta$ maps to build an interleaving. 

To that end, consider the diagram 
\begin{equation}
\label{eq:EquivMetrics_Dgm1}
\begin{tikzcd}[column sep = huge]%
       & 
       & \color{violet} S_{2\delta}^{2m\delta}(G,f)
       \\
       & \color{violet} S_{\e}^{m\e} (G,f) 
            \ar[ur,"\rho"]
       & S_{2\e}^{(m+m')\e}(G,f) 
           \ar[u, "\nu", red]
       \\
       \color{violet} (G,f) 
       \ar[r,"\rho", dashed]  
       \ar[dr, blue, "\alpha"',very near start]
       \ar[ur,"\rho"]
       & 
       S_{\e}^{m'\e} (G,f) 
            \ar[r,"\rho",dashed ] 
            \ar[u, "\nu", red]
            \ar[ur,"\rho", dashed]  
            \ar[dr, blue, "{S_{\e}^{m'\e}[\alpha]}"',very near start]
       & 
       S_{2\e}^{2m'\e}(G,f)
           \ar[u, "\nu", red]
       \\
       \color{violet} (H,h) 
       \ar[r,"\rho", dashed]  
       \ar[dr,"\rho"]
       \ar[ur, red, "\beta"',crossing over,very near end]
       & 
       S_{\e}^{m'\e} (H,h) 
            \ar[r,"\rho", dashed] 
            \ar[d, "\nu", blue]
            \ar[dr,"\rho", dashed]  
            \ar[ur, red, "{S_{\gamma}^{m'\gamma}[\beta]}"',very near end, crossing over]
       & 
       S_{2\e}^{2m'\e}(H,h)
           \ar[d, "\nu", blue]
       \\
       & \color{violet} S_{\e}^{m\e} (H,h) 
            \ar[dr,"\rho"]
       & S_{2\e}^{(m+m')\e}(H,h) 
           \ar[d, "\nu", blue]
       \\
       & %
       & \color{violet} S_{2\e}^{2m\e}(H,h) 
\end{tikzcd}
\end{equation}
Note that some quick calculations ensure that all the denoted $\rho$ maps exist because all slopes (viewed as in [fig whatever]) are at most 1. 
All triangles in the top and bottom half commute because of \cref{lem:triangles}. 
The parallelograms commute because $\rho$ is a restriction of the $\eta$ maps along the inclusions $\nu$. 
The middle strip commutes because $\alpha$ and $\beta$ consitute an interleaving.

So, let $\color{blue}\alpha':= \nu \alpha $ and $\color{red}\beta' = \nu \beta $ be the compositions
\[
(G,f) \xrightarrow{\alpha} S_{\epsilon}^{m'\epsilon}(H,h) \xrightarrow{\nu} S_{\epsilon}^{m\epsilon}(H,h) 
\qquad
(H,h) \xrightarrow{\beta} S_{\epsilon}^{m'\epsilon}(G,f) \xrightarrow{\nu} S_{\epsilon}^{m\epsilon}(G,f). 
\] 
We now ensure that 
$S_{\e}^{m\e}[\alpha'] = \nu \circ S_{\e}^{m'\e}[\alpha]$.
To see this consider the diagram 
\begin{equation*}
\begin{tikzcd}
G 
    \ar[rrr,"\rho"]
    \ar[dr,"\alpha"']
    \ar[dd, bend right, "\alpha'", violet]
&&&
S_\e^{m'\e}G
    \ar[dl,"{S_\e^{m'\e}[\alpha]}"]
    \ar[dd, bend left, "{S_\e^{m'\e}[\alpha']}", violet]
\\
&
S_\e^{m'\e}H 
    \ar[r,"\rho"]
    \ar[dl, "\nu"']
    \ar[drr, "\rho"]
& 
S_{2\e}^{2m'\e}H
    \ar[dr, "\nu"]
\\
S_\e^{m\e}H 
    \ar[rrr,"\rho"]
& &&
S_{2\e}^{2m\e}H
\end{tikzcd}
\end{equation*}
where the rightmost triangle is what we need to show commutes. 
The leftmost triangle is the definition of $\alpha'$. 
The top square and the outer face commute by \cref{lem:build_nat_trans}. 
The remaining two triangles in the middle-bottom commute by \cref{lem:triangles}. 
Thus, the rightmost triangle commutes as desired, and the proof that 
$S_{\e}^{m\e}[\beta'] = \nu \circ  S_{\e}^{m'\e}[\beta]$. 
Hence, $\alpha'$ and $\beta'$ constitute an $S_\e^{m'\e}$ interleaving, so $d_I^m \leq d_I^{m'}$.

\vspace{.5ex}
\noindent\textbf{Second inequality.}
For the other direction, we need to show 
\[
d_I^{m'}((G,f), (H,h)) 
\leq \frac{1-m}{1-m'}  d_I^{m}((G,f), (H,h)). 
\]
Assume we have an $\e$-interleaving using $m$ given by 
\begin{align*}
   \alpha\from  &(G,f) \to S_{\e}^{m\e}(H,h)\\
   \beta\from  &(H,h) \to S_{\e}^{m\e}(G,f)
\end{align*} 
We will use these maps to build a $\delta = \frac{1-m}{1-m'}$ interleaving using $S_{\e}^{m'\e}$.
Consider the diagram of \cref{fig:SecondInequalityDiagrams-Half},
where all $\rho$ maps exist as all slopes are less than 1, and the diagram commutes because of \cref{lem:triangles}. 

\begin{figure}%
    \centering
    \includegraphics{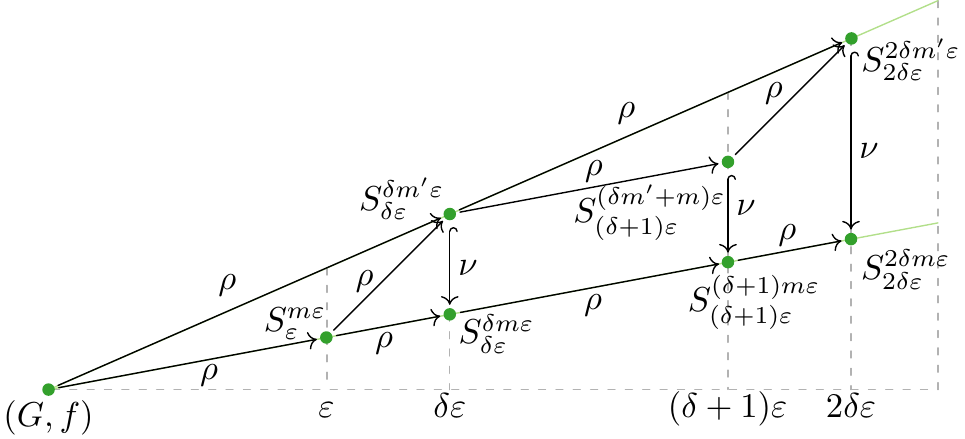}
    \caption{A diagram in the $\e$-$\tau$ plane used in the proof of \cref{thm:EquivalentMetrics}.}
    \label{fig:SecondInequalityDiagrams-Half}
\end{figure}

We now use this diagram to build $\alpha'$ and $\beta'$ as in the first half of the proof. 
Consider the diagram 
\begin{equation}
\label{eq:SecondInequalityDiagram}
\begin{tikzcd}[column sep = huge, row sep = huge]
    &&&&
    S_{2\delta\e}^{2m'\delta\e}
        \ar[dd,"\nu",hook]
    \\
    &
    & 
    S_{\delta \e}^{m'\delta\e} 
        \ar[d,"\nu",hook]
        \ar[r]
        \ar[urr]
        \ar[dddr,dashed,blue,"{ {S_{\delta \e}^{\delta m' \e}[\alpha]}}", very near start]
    &
    S_{(\delta+1)\e}^{(\delta m' + m)\e}
    \ar[d,"\nu",hook]
    \ar[ur]
    \\
    (G,f) 
        \ar[urr]
        \ar[r]
            \ar[dr,blue,"\alpha",very near start ]
    &
    S_\e^{m\e}
        \ar[r]
        \ar[ur]
    & 
    S_{\delta \e}^{m\delta\e} 
        \ar[r]
            \ar[dr, blue, "{S_{\delta \e}^{\delta m \e}[\alpha]}"', very near start]
    &
    S_{(\delta+1)\e}^{m(\delta+1)\e}
        \ar[r]
    &
    S_{2\delta\e}^{2m\delta\e}
    \\
    (H,h) 
        \ar[drr]
        \ar[r]
            \ar[ur,"\beta"',red, very near start, crossing over]
    &
    S_\e^{m\e}
        \ar[r]
        \ar[dr]
    & 
    S_{\delta \e}^{m\delta\e} 
        \ar[r]
        \ar[ur, red, "{ S_{\delta \e}^{\delta m \e}[\beta]  }", very near start]
    &
    S_{(\delta+1)\e}^{m(\delta+1)\e}
        \ar[r]
    &
    S_{2\delta\e}^{2m\delta\e}
    \\
    &
    & 
    S_{\delta \e}^{m'\delta\e} 
        \ar[u,"\nu",hook]
        \ar[r]
        \ar[drr]
        \ar[uuur,dashed,red, "{ { S_{\delta \e}^{\delta m' \e}[\beta]}  }"', very near start]
    &
    S_{(\delta+1)\e}^{(\delta m' + m)\e}
    \ar[u,"\nu",hook]
    \ar[dr]
    \\
    &&&&
    S_{2\delta\e}^{2m'\delta\e}
        \ar[uu,"\nu",hook]
\end{tikzcd}
\end{equation}
where we drop the $(G,f)$ notation on the top half and the $(H,h)$ on the bottom half for ease of reading.
The dashed functions are restrictions of their solid counterparts because $T^\cdot[\cdot]$ is defined as such; i.e.~$S_{\delta \e}^{ m'\delta\e}[\alpha] = T^\tau[S_{\delta \e}^{m \delta \e}[\alpha]]$ for $\tau = (m'-m) \delta \e$. 
Thus, the squares 
\begin{equation*}
\begin{tikzcd}%
    S_{\delta \e}^{m'\delta\e} (G,f)
            \ar[r, blue, "{ { S_{\delta \e}^{\delta m' \e}[\alpha]}}",  dashed]
            \ar[d,hook, "\nu"']
    &     
    S_{(\delta+1)\e}^{(\delta m' + m)\e}(H,h)
            \ar[d,hook, "\nu"]
    &
    S_{\delta \e}^{m'\delta\e} (H,h)
            \ar[r, red, "{ {S_{\delta \e}^{\delta m' \e}[\beta]}}",  dashed]
            \ar[d,hook, "\nu"']
    &     
    S_{(\delta+1)\e}^{(\delta m' + m)\e}(G,f)
            \ar[d,hook, "\nu"]
\\
    S_{\delta \e}^{m\delta\e} (G,f)
            \ar[r, blue, "{S_{\delta \e}^{\delta m \e}[\alpha]}"', ]
    &     
    S_{(\delta+1)\e}^{m(\delta+1)\e}(H,h)
    &
    S_{\delta \e}^{m\delta\e} (H,h)
            \ar[r, red, "{ S_{\delta \e}^{\delta m \e}[\beta]}"', ]
    &     
    S_{(\delta+1)\e}^{m(\delta+1)\e}(G,f)
\end{tikzcd}
\end{equation*}
commute. 
The middle strip commutes because $\alpha$ and $\beta$ constitute an interleaving. 
Thus the diagram of \cref{eq:SecondInequalityDiagram} commutes. 

Define $\color{blue}\alpha': = \rho \alpha $ and $\color{red}\beta':= \rho \beta $ by the compositions
\[
G \xrightarrow{\alpha} S_{\epsilon}^{m\epsilon}H \xrightarrow{\rho} S_{\delta\epsilon}^{m'\delta\epsilon}H 
\qquad
H \xrightarrow{\beta} S_{\epsilon}^{m\epsilon}G \xrightarrow{\rho} S_{\delta\epsilon}^{m'\delta\epsilon}G. 
\]
To ensure that $\alpha'$ and $\beta'$ constitute an interleaving, we need to ensure that 
\begin{equation*}
S_{\delta \e}^{m'\delta\e}[\alpha'] = \rho {S_{\delta \e}^{m'\e\delta }[\alpha]} 
\qquad \text{ and } \qquad 
S_{\delta \e}^{m'\delta\e}[\beta'] = \rho {S_{\delta \e}^{m'\e\delta }[\beta]}. 
\end{equation*}
We only check the first because the second is symmetric. 
Consider the following diagram
\begin{equation*}
\begin{tikzcd}
    (G,f) 
        \ar[r,"\rho"]
        \ar[rr, bend left, "\rho"]
        \ar[d, blue, "\alpha"]
        \ar[ddd, violet, "\alpha'"', bend right = 75]
    & S_{\delta \e}^{m\delta \e}(G,f)
        \ar[d, blue, "{S_{\delta \e}^{\delta m \e}[\alpha]}"]
    & S_{\delta \e}^{m'\delta \e}(G,f)
        \ar[l, hook, "\nu"]
        \ar[d, blue, "{S_{\delta \e}^{\delta m \e}[\alpha]}"]
        \ar[ddd, violet, "{S_{\delta \e}^{m'\delta \e}[\alpha']}", bend left = 75]
    \\
    S_{\e}^{m\e}(H,h) 
        \ar[r, "\rho"]
        \ar[dd, "\rho"]
        \ar[dr, "\rho"]
    & S_{(\delta + 1)\e}^{(\delta + 1)m\e}(H,h) 
    & S_{(\delta + 1)\e}^{(\delta m'  + m)\e}(H,h) 
        \ar[l, hook, "\nu"]
        \ar[dd, "\rho"]
    \\
    & S_{\delta \e}^{m\delta \e}(H,h)
        \ar[u, "\rho"]
    \\
    S_{\delta \e}^{m'\delta \e}(H,h)
        \ar[uurr, "\rho", bend right]
        \ar[ur, hook, "\nu"]
        \ar[rr, "\rho", bend right]
    & 
    & S_{2\delta \e}^{2m'\delta \e}(H,h)
\end{tikzcd}
\end{equation*}
which is the relevant portion of the 3D diagram of \cref{fig:SecondInequality3d}. 
We need to ensure the right most triangle commutes. 
We can check that all faces of the diagram only involving $\rho$ and $\nu$ maps commute by \cref{lem:triangles,lem:squareCommutes}. 
The left-most triangle commutes by the definition of $\alpha'$. 
the square whose upper left corner is $(G,f)$ and the outside face commute by \cref{lem:build_nat_trans}. 
The remaining square commutes by \cref{prop:mapRestrictsTruncation}. 
Thus, the rightmost triangle commutes as desired.

Finally, this combined with  \cref{eq:SecondInequalityDiagram} implies that 
\begin{equation*}
\begin{tikzcd}[column sep = huge, row sep = huge]
(G,f) 
    \ar[r, "\rho"]
    \ar[dr, "\alpha'"', very near start]
&
S_{\delta \e}^{m'\delta \e} (G,f) 
    \ar[r, "\rho"]
    \ar[dr, "{S_{\delta \e}^{m'\delta \e}[\alpha']}"', very near start]
& 
S_{2\delta \e}^{2m'\delta \e} (G,f) 
\\
(H,h) 
    \ar[r, "\rho"]
    \ar[ur, "\beta'", very near start]
&
S_{\delta \e}^{m'\delta \e} (H,h) 
    \ar[r, "\rho"]
    \ar[ur, "{S_{\delta \e}^{m'\delta \e}[\beta']}", very near start]
& 
S_{2\delta \e}^{2m'\delta \e} (H,h) 
\end{tikzcd}
\end{equation*}
commutes.
So, since $\alpha'$ and $\beta'$ constitute an interleaving with $S_{\delta \e}^{m'\delta \e}$, we have that 
$
d_I^{m'}\leq \frac{1-m}{1-m'}  d_I^{m} 
$,
concluding the proof.
\end{proof}

\begin{figure}
    \centering
    \includegraphics{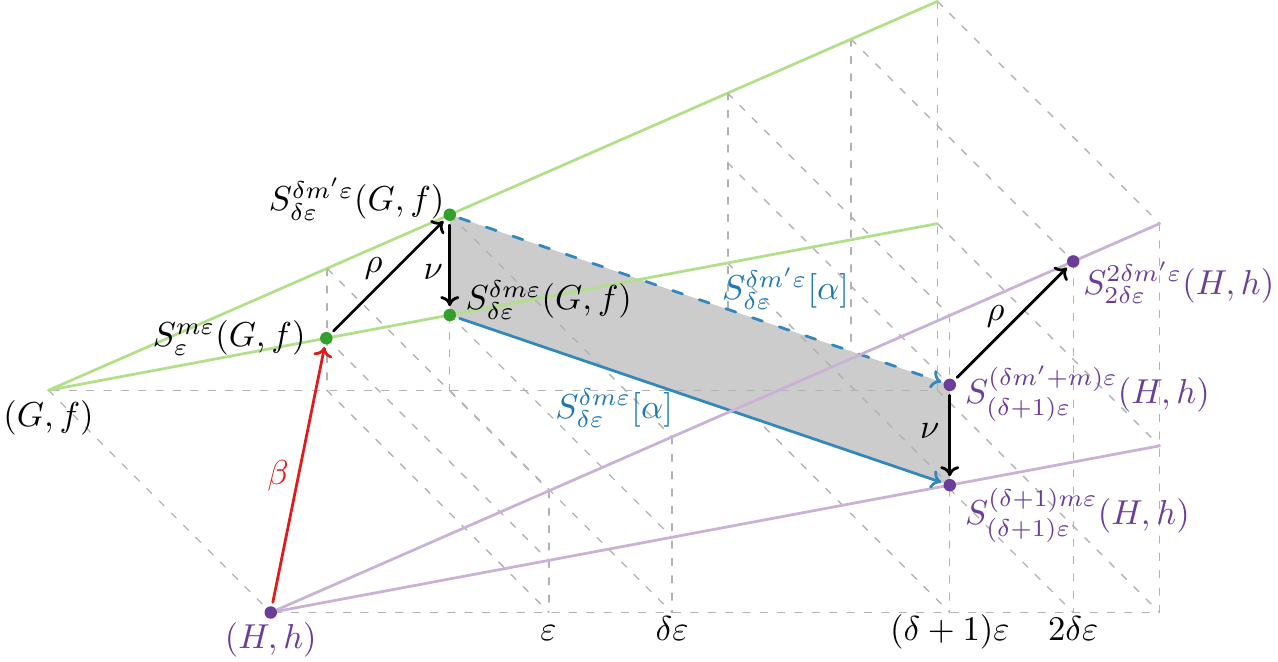}
    \caption{Another commutative diagram used in the proof of \cref{thm:EquivalentMetrics}.
    }
    \label{fig:SecondInequality3d}
\end{figure}
\section{Conclusion and discussion}
\label{sec:Conclusions}

Our primary aim has been to introduce the concept of truncated smoothing and establish properties and connections of this operation.
We have several reasons for considering this as a similarity measure on Reeb graphs.
First, it has potential for providing bounds for the stable interleaving distance via the equivalence of metrics.
Second, we came to this definition while investigating drawings of Reeb graphs and when  planarity is achievable (that, is whether a Reeb graph has a planar drawing which respects the function in the $y$-coordinate).
In a subsequent paper, we will show that while traditional smoothing does not maintain planarity, the truncated smoothing does for $\e \leq \tau \leq 2\e$.

We suspect additional potential applications of truncated smoothing in comparing geometric or planar graphs, since it simplifies the graph's topology (via smoothing) without suffering from extensive expansion of the co-domain or destruction of desirable combinatorial properties like level planarity.
Truncated smoothing also allows for interesting manipulation of the extended persistence diagram of the Reeb graph and computation of morphs between Reeb graphs; again, we defer details to future work, as a full classification of that manipulation is necessary.

We suspect that the loss of stability discussed in \cref{ss:propertiesOfMetric} is not as dire as it seems.
If nothing else, \cref{prop:stabilityWIthConstant} gives a Lipschitz constant in advance dependent only on $m$, so
it is possible to upper bound the difference using these new interleaving distances.

While we are able to connect our collection of metrics to several Reeb metrics (\cref{Prop:FD}), we have not investigated further connections to other metrics as of yet.
One particularly interesting future direction is to determine whether this collection of metrics provides results related to strong equivalence between the interleaving distance and the universal distance of \cite{Bauer2020a}.
Perhaps this broader collection of metrics will help to provide stronger bounds between the various metrics on Reeb graphs, since strong equivalence with one is strong equivalence with all.

\printbibliography

\end{document}